%% file: thermalarxiv-rev2.tex
\documentclass[a4paper,11pt]{article}
\usepackage{pgf}
\usepackage{url}
\usepackage{enumerate}
\usepackage{amsthm,amssymb,amsmath,amscd}
\usepackage{eucal}
\usepackage{mathrsfs}
\usepackage{mathtools}
\usepackage{upgreek}
\usepackage{dsfont}
\usepackage{yfonts} 
\usepackage{subfigure}
\usepackage{lmodern}
\usepackage[T1]{fontenc}
\usepackage{hyperref}
\usepackage{bbm}
\usepackage{geometry}
\usepackage{array}
\usepackage{amsthm}
\usepackage{booktabs}
\usepackage{float} 
\usepackage[active]{srcltx}
\usepackage{authblk}

\newcommand{\RR}{\mathbbm R}
\newcommand{\CC}{\mathbbm C}
\newcommand{\NN}{\mathbbm N}
\renewcommand{\vec}{\boldsymbol}
\newcommand{\g}{G}
\DeclareMathOperator{\supp}{supp}

\newtheorem{proposition}{Proposition}
\newtheorem{theorem}{Theorem}
\newtheorem{definition}{Definition}
\DeclarePairedDelimiter\abs{\lvert}{\rvert}%
\newcommand{\normord}[1]{\mathinner{\mathopen{\boldsymbol:}#1\mathclose{\boldsymbol:}}}
\renewcommand{\abs}[1]{\mathinner{\left\lvert#1\right\lvert}}

\newcommand{\WF}{\textup{WF}}
\newcommand{\fdiff}[1]{\rho^{#1}_+}
\newcommand{\pdiff}[1]{\rho^{#1}_-}
\newcommand{\fdifft}[1]{\rho'^{#1}_+}
\newcommand{\pdifft}[1]{\rho'^{#1}_-}

 \usepackage{textcomp}
 
 \newcommand\quotient[2]{
        \mathchoice
            {
                \text{\raise1ex\hbox{$#1$}\Big/\lower1ex\hbox{$#2$}}%
            }
            {
                #1\,/\,#2
            }
            {
                #1\,/\,#2
            }
            {
                #1\,/\,#2
            }
    }
\begin{document}
%
%
%
%
%
%
%
%
%

\title{Construction of KMS States in Perturbative QFT and Renormalized Hamiltonian Dynamics}
\author[]{Klaus Fredenhagen}
\author[]{Falk Lindner}
\affil[]{II. Institute for Theoretical Physics, University of Hamburg\\{\small{\texttt{klaus.fredenhagen@desy.de}, \texttt{falk.lindner@desy.de}}}}
\date{}
\maketitle
\vspace{-2em}
\begin{abstract}
We present a general construction of KMS states in the framework of perturbative algebraic quantum field theory (pAQFT). Our approach may be understood as an extension of the Schwinger-Keldysh formalism. We obtain in particular the Wightman functions at positive temperature, thus solving a problem posed some time ago by Steinmann \cite{steinmannfinite}. The notorious infrared divergences observed in a diagrammatic expansion are shown to be absent due to a consequent exploitation of the locality properties of pAQFT. To this avail, we introduce a novel, Hamiltonian description of the interacting dynamics and  find, in particular, a precise relation between relativistic QFT and rigorous quantum statistical mechanics.
\begin{center}
	\emph{Dedicated to the memory of  Othmar Steinmann $^* 27.11.1932   \quad ^\dagger 11.03.2012$}
\end{center}
\end{abstract}
\tableofcontents

\section{Introduction}
According to the standard model of cosmology, the early universe was for some time in an equilibrium state with high temperature. At these temperatures, matter has to be described by an interacting quantum field theory. Also a variety of other phenomena are studied within the framework of quantum field theory at positive temperature. This includes in particular the thermodynamics of quark-gluon plasmas 
where one expects a phase transition from a confined to a deconfined phase as predicted by lattice QCD. For other examples and an introduction to the methods in this field we refer to \cite{zinn}. There are, however, several problems with this project:

First of all, there is no generally accepted concept of temperature in an expanding universe (for a discussion see e.g.\  \cite{buchsolv}). But curvature is, in the relevant period, already small compared with the heuristically assigned temperature, hence an approximation by a positive temperature state on Minkowski spacetime seems to be meaningful as long as the considered time scales are sufficiently small.

Once one accepts this ansatz one has to construct interacting quantum field theory at positive temperature on Minkowski spacetime. Surprisingly, this problem turns out to be much harder than one might have expected before. Standard quantum field theory is based on an expansion in Feynman graphs. If one replaces in these graphs the Feynman propagator by its positive temperature analogue one finds quite a number of infrared divergences whose systematic removal is not at all obvious.

Several approaches to this problem have been proposed in the literature, in particular Umezawa's thermal field theory \cite{umezawa,matsumoto}, the euclidean approach by Matsubara \cite{matsubara,bellac} and the Schwinger Keldysh path integral \cite{schwinger,keldysh}. A review of these methods may be found in a paper by Landsman and van Weert \cite{lvw} where also the relation to rigorous quantum statistical mechanics is discussed.

The difficulties may be traced back to the change in the behavior at large times induced by the interaction. In the vacuum sector of quantum field theory a corner stone is the LSZ asymptotic condition which postulates that the interacting field behaves at large times as a free field. This behavior can actually be derived from the general axioms of quantum field theory (Wightman axioms or Haag-Kastler axioms) in the presence of isolated mass shells in the energy momentum spectrum \cite{haagbuch} (for a recent  improvement, see \cite{dybalski}). It may be interpreted as a consequence of the fact that (stable) particles are far from each other at large times such that their interaction can be neglected.

At positive temperature this assumption clearly is no longer satisfied. Moreover, even the existence of stable particles (defined as eigenstates of the mass operator) is incompatible with interaction, as was shown by Narnhofer, Requardt and Thirring \cite{nrt}. The presumably correct behavior at large times has been derived by Bros and Buchholz for $\phi^4$ under some plausible assumptions \cite{bb}.

Equilibrium states at positive temperature can be characterized by the KMS condition. This condition is valid for Gibbs states and holds also in the thermodynamic limit  \cite{hhw}. In quantum statistical mechanics
much is known about these states (see e.g.\  \cite{brob} or, for a recent overview \cite{sakai}). A construction was also possible  in some superrenormalizable two-dimensional relativistic quantum field theories by means of functional integral methods \cite{hoegh,gerard-jaekel1,gerard-jaekel2}. 

In the case of perturbative quantum field theory in four spacetime dimensions, Steinmann described  in two seminal papers the perturbation expansion of Wightman functions at zero \cite{steinmannvac} and nonzero \cite{steinmannfinite} temperature. While at zero temperature the existence and uniqueness of the series could be established, at positive temperature only uniqueness could be proved whereas the existence was unclear due to a number of infrared divergences. A cancellation of these divergences was not excluded, but was not visible. Actually, some of these divergences were discussed in the mentioned review of Landsman and van Weert and there shown to cancel. A general proof that all arising divergences cancel does not seem to exist, to the best of our knowledge.        

In euclidean field theory at positive temperature the problem is absent, and a rigorous perturbative construction was performed by Kopper, M\"uller and Reisz \cite{kopper}. There is, however, no reconstruction theorem comparable to the Osterwalder-Schrader Theorem \cite{os} which is valid at nonzero temperature and which covers the case of perturbative interacting quantum field theory in four dimensional Minkowski space\footnote{The existing theorems (see \cite{BFro} for an account) make assumptions on the existence of the algebra at time zero which are satisfied in $P(\phi)_2$-theories (see \cite{JR}), but no longer in more singular theories as interacting quantum field theory in 4 dimensions.}.

In this paper we develop a new approach to the problem, based on perturbative algebraic quantum field theory (pAQFT), see \cite{paqftbuch,frerejz} for an introduction. pAQFT provides a state independent construction of the local algebras of observables, and the problem of field theory at positive temperature is then reduced to the construction of a KMS state on this algebra.

In pAQFT one first introduces a spacetime cutoff by multiplying the interaction Lagrangian with a test function $g$ of compact support. One then constructs time ordered products of all local fields by the algebraic version of causal perturbation theory which was mainly developed for quantum field theory on curved spacetimes. In terms of the time ordered products one can define the interacting field inside the algebra of the free field, thereby avoiding the consequences of Haag's theorem on the nonexistence of the interaction picture in field theory. The vacuum is then constructed in the following way: one considers the expectation values of products of interacting fields in the vacuum state of the free theory and studies their behavior in the limit when $g$ tends to 1 (adiabatic limit, see e.g.\  \cite{eg1}).

A corresponding method seems to fail at positive temperature due the persistent influence of the interaction at asymptotic times. We therefore use a different approach. We exploit the fact that the interacting theory satisfies the time-slice axiom \cite{cf} so that it suffices to look at interacting fields within a finite time interval. We then exploit the causality properties of causal perturbation theory and construct the interacting field in the adiabatic limit within the algebra of the free field associated to a somewhat larger time interval, a method first described by Hollands and Wald \cite{hwexistence}. The time evolution of the interacting field is then related to the free time evolution by a co-cycle which is generated by a time-averaged interacting Hamiltonian density.

We are now in a position to apply standard methods of rigorous quantum statistical mechanics in order to construct a KMS state of the interacting theory as a perturbation of a KMS state of the free theory. We use the expansion in terms of truncated functions developed by Araki \cite{arakikms,brob} and exploit the KMS condition of the free theory and the spatial decay of correlations and finally succeed in an explicit construction of an equilibrium state at positive temperature.
\section{Perturbed dynamics}\label{sec:2}
A characteristic difficulty of quantum field theory is the singular behavior of the relevant interactions. As a matter of fact, the interaction Hamiltonian, formally written as a spatial integral of the Hamiltonian density,
\begin{align}
  H_I=\int d^3\vec{x}\;\mathcal{H}_I(0,\vec{x})	\label{eq:interactionhamilton}
\end{align}
suffers from several problems:
\begin{enumerate}

  \item the Hamiltonian density involves pointlike products of fields; this problem can be solved by normal ordering, but it has to be done in a state independent way \cite{hwlocal,bfv}.

  \item in four dimensional spacetime all normal ordered polynomials of the free field of degree larger than 1 cannot be restricted to a spacelike surface as operator valued distributions.\label{2}

  \item  the integral over all space typically does not exist. \label{3}

\end{enumerate}
Problem \ref{3} is related to Haag's Theorem \cite{haagthm,haagthm2} and may be avoided by looking at the induced derivation on the algebra of local observables $A$,
  \[\delta_I(A)=i\int d^3\vec{x}\;[\mathcal{H}_I(0,\vec{x}),A]\ .\] 
Problem \ref{2} however reenters if a perturbation expansion for the full dynamics is based on this derivation. Moreover, in general the algebra of canonical commutation relations can no longer be used for the interacting theory due to a nontrivial field strength renormalization.
  
For these reasons, a perturbative expansion as in quantum statistical mechanics \cite{brob} was considered  not to be possible for quantum field theory, and one developed other formalisms which, however, have also problems. In particular, the infrared problems stated above have not been solved therein.

A method which is relatively near to quantum mechanics is based on the formal expansion of the time evolution operator of the interaction picture in terms of time ordered products of interaction densities,
\begin{align*}
  U(t,s)&=\sum_{n=0}^{\infty} (-i)^n\int_{s}^tdt_1\int_s^{t_1}dt_2\cdots \int_s^{t_{n-1}}dt_n  \int d^3\vec{x}_1\dots d^3\vec{x}_n \, \mathcal{H}_I(t_1,\vec{x}_1)\dots\mathcal{H}_I(t_n,\vec{x}_n)\\	
  &=\sum_{n=0}^\infty \frac{(-i)^n}{n!}\int_{s}^t dt_1\cdots \int_{s}^t dt_n \int d^3\vec{x}_1\cdots d^3\vec{x}_n \;T\mathcal H_I(t_1,\vec{x}_1)\cdots \mathcal H_{I}(t_n,\vec{x}_n)
\end{align*}
where $T$ denotes time ordering. If one replaces in this formula  the $n$-fold integral over the time-slice $[s,t]\times \RR^3$ by the integral over a test function with compact support one is left with the problem to define time ordered products as  operator valued distributions. For $n=1$, time ordering has no effect, and, indeed by the G\aa{}rding-Wightman theorem \cite{gaarding}, the normal ordered polynomials of the free field are well-defined operator valued distributions.

The next step is to define the time ordered product for $n>1$. By definition, the time ordered product is well defined at non-coinciding points where it is the operator product in the appropriate order. It remains the problem to extend the time ordered products to all points, in the sense of operator valued distributions.

This problem, originally posed by St\"uckelberg and Bogoliubov, was solved by Epstein and Glaser \cite{eg1}. They were able to show that these extensions always exist and are unique up to finite renormalizations; the ambiguity corresponds exactly to the usual freedom of choosing renormalization conditions. One ends up with the formal $S$-matrix
\[S(g)=\sum_{n=0}^{\infty}\frac{(-i)^n}{n!}
\int d^4 x_1\dots d^4 x_n \, T\mathcal{H}_I(x_1)\dots\mathcal{H}_I(x_n)g(x_1)\dots g(x_n)\]
where $g\in\mathcal D(\RR^4)$ is a test-function with compact support. $S(g)$ is the generating functional of time ordered products of the interaction density.

How is $S(g)$ related to observable quantities? First of all, it is expected that in the adiabatic limit $g\to 1$ the formal $S$-matrix $S(g)$ tends to the physical $S$-matrix in  the Fock space representation of the theory, since in this limit it agrees with Dyson's formula. For this to hold one has to choose renormalization conditions such that the renormalized mass and the field strength renormalization coincides with the free theory \cite{eg2}. For positive temperature this is of no use because of the different asymptotic behavior.

Another important relation to observables was discovered by Bogoliubov. Namely, let $A_i$, $i=1,\ldots,N$ be local fields of the theory with $A_1=\mathcal H_I$. Provided all time ordered products of these fields have been fixed, one may introduce the formal $S$-matrix 
\begin{equation}
   S(f)=\sum_{n=0}^\infty \frac{(-i)^n}{n!}\int d^4x_1\dots d^4x_n\sum_{k_1,\dots,k_n=1}^{N}TA_{k_1}(x_1)\dots A_{k_n}(x_n)f_{k_1}(x_1)\dots f_{k_n}(x_n) 
\end{equation}
with $f\in\mathcal D(\RR^4,\RR^{N})$. $S(f)$ may be interpreted as the $S$-matrix with the interaction  $A(f)=\int d^4x\sum A_k(x)f_k(x)$. 
Actually, as shown in \cite{dfr}, one can perform all these constructions within an abstract *-algebra $\mathfrak{A}$. $\mathfrak{A}$ is isomorphic to the algebra of smeared normal ordered products of the free field on Fock space,
where the smearing involves not only test functions but also certain distributions with compact support. For a spacetime region $\mathcal O$ we define $\mathfrak{A}(\mathcal O)$ as the subalgebra of $\mathfrak{A}$ of those smeared normal ordered products where the smearing is restricted to a compact subset of $\mathcal O$. The formal $S$-matrices are then unitary elements of the *-algebra $\mathfrak{A}[[\lambda]]$ of formal power series with values in 
$\mathfrak{A}$. We will often suppress the formal parameter in what follows keeping in mind that our treatment of interactions is always in the sense of formal perturbation theory.

 A crucial object in the description of interacting fields is the so-called \emph{relative $S$-matrix}
\begin{align*}
	S_g(f)&:= S(g)^{-1}S(g+f)\ ,\ f,g\in\mathcal D(\RR^4,\RR^{N})\ .
\end{align*}

Bogoliubov discovered \cite{bogolubovintro} that one can define the interacting field by the formula
\begin{align}
  \left[A_i(x)\right]_{g}:=\frac{\delta}{\delta f_i(x)}\Big|_{f=0}S_g(f) \label{eq:bogoliubov}
\end{align}
as a formal power series of  operator valued distributions on the Fock space of the free theory. Furthermore he also proved that the relative $S$-matrix $S_g(f)$ depends only on the restriction of $g$ to the causal past $J_-(\supp f)$ of the support of $f$ (where $J_{\pm}(\mathcal O)$ denote the closures of the regions which can be reached from the spacetime region $\mathcal O$ by a future or past  directed causal path, respectively).
This is a consequence of the improved causal factorization condition \cite{eg1}
\begin{align}
   S\left(f + g +h\right)= S\left(f + g \right) S\left(g \right)^{-1} S\left( g + h\right),\label{eq:causalfactor0}
\end{align}
which holds (irrespective of $\supp(g)$ and the choice of the local fields in $A$) if $\supp(f)$ does not intersect the past of $\supp(h)$ . The formula is equivalent to each of the following relations of the relative $S$-matrix 
\begin{align}
	S_g(f+h)&=S_g(f)S_g(h)\label{eq:causalfactor1}\\
	S_{g+f}(h)&=S_{g}(h)\label{eq:causalfactor2}\\
	S_{g+h}(f)&=S_g(h)^{-1}S_g(f)S_g(h)=S_g(h)^{-1}S_g(f+h)\label{eq:causalfactor3}
\end{align}
if $\supp(f)$ does not intersect with the past of $\supp(h)$. The relations \eqref{eq:causalfactor1} and {\eqref{eq:causalfactor2} are heavily exploited in causal perturbation theory on Minkowski spacetime \cite{eg1} and led to the development of causal perturbation theory on curved spacetimes, see \cite{bfbg}.

The retarded interacting field $[A_i]_g$ with interaction $A(g)$, defined by Boboliubov's formula \eqref{eq:bogoliubov}, coincides at points $x\not\in J_+(\supp g)$ with the field $A_i$ of the free theory. The singularities of the naive interaction picture are avoided in this framework, since the interaction is switched on in a smooth way described by the spacetime cutoff $g\in\mathcal  D$, compared to the instantaneous switching at time zero. For a discussion on this topic, see \cite{scharf}.

The improved causality condition \eqref{eq:causalfactor0} (which is implied by the causal factorization of time ordered products and by the definition of the formal $S$-matrix as a formal power series) turns out to be crucial for the construction of the algebra of interacting fields. Namely, let $\mathfrak{A}_{g}(\mathcal{O})$ be the algebra generated by the relative $S$-matrices $S_g(f)$ with $\supp f\subset\mathcal{O}$ where $\mathcal{O}$ is a relatively compact region of spacetime. We know that $S_g(f)$ depends only on the behavior of $g$ within $J_-(\mathcal{O})$.  But, as first observed in \cite{slavnov} and systematically exploited in \cite{bfbg}, the dependence on $g$ in that part of the past which is outside of $J_+(\mathcal{O})$ is via a unitary transformation which is independent of $f$, i.e. if $g'$ coincides with $g$ on a neighborhood of $J(\mathcal O):=J_+(\mathcal O)\cap J_-(\mathcal O)$, then there exists a unitary $W(g',g)\in\mathfrak{A}[[\lambda]]$ such that
\begin{equation}
     S_{g'}(f)=W(g',g)S_g(f)W(g',g)^{-1}
\end{equation}
holds for all $f$ with $\supp f\subset\mathcal O$. Thus the algebra $\mathfrak{A}_{g}(\mathcal{O})$ is, up to isomorphy, uniquely determined by the restriction of $g$ to the causal completion $J(\mathcal O)$ of $\mathcal{O}$. We denote this abstract algebra by $\mathfrak{A}_{[g]}(\mathcal O)$ where $[g]\equiv[g]_{\mathcal O}$ is the class of all test functions which coincide with $g$ on a neighborhood of $J(\mathcal O)$. One then can insert for $g$ a smooth function $\g$ without restrictions on the support.  
The algebra  $\mathfrak{A}_{[\g]}(\mathcal O)$ is generated by maps
\begin{equation}
    S_{[\g]}(f):[\g]_{\mathcal O}\to\mathfrak A\ ,\ g\mapsto S_{g}(f)\ ,\ \supp f\subset\mathcal O\ ,
\end{equation} 
with pointwise algebraic operations.

For $\mathcal O_1\subset \mathcal O_2$ one obtains a natural embedding $i_{\mathcal O_2,\mathcal O_1}:\mathfrak{A}_{[\g]}(\mathcal O_1)\to\mathfrak{A}_{[\g]}(\mathcal O_2)$
by restricting the maps $S_{[\g]}(f)$ from $[\g]_{\mathcal O_1}$ to $[\g]_{\mathcal O_2}$.
As shown in \cite{bfbg}, 
this allows to construct the adiabatic limit $\g=(1,0,\dots,0)$ for the net of local observables in the sense of the Haag-Kastler axioms (algebraic adiabatic limit).   

\section{Time averaged Hamiltonian description of the dynamics}
In this section we introduce a version of the adiabatic limit which goes back to Hollands and Wald \cite{hwexistence} and which allows a discussion of the interacting dynamics in close analogy with the Hamiltonian formalism but where the restriction to a Cauchy surface is replaced by the restriction to a time-slice, thus avoiding the UV divergences accompanying the Hamiltonian formalism in four dimensional quantum field theory.

The time-slice axiom \cite{haagschroer} of quantum field theory is a weak form of the determination of the theory by initial conditions. It is also called primitive causality. In the algebraic framework it says that the algebra of a globally hyperbolic subregion $\mathcal O$ of spacetime is equal to the algebra of $\mathcal O_1\subset\mathcal O$, if $\mathcal O_1$ is a neighborhood of some Cauchy surface of $\mathcal O$. The time-slice axiom holds in perturbative QFT \cite{cf}.
It implies that the algebra of local observables of the interacting theory is, for every $\epsilon>0$, generated by the relative $S$-matrices $S_{[\g]}(f)$ with $\supp f\subset\Sigma_{\epsilon}=\{(t,\vec{x})|-\epsilon<t<\epsilon\}$. 

Now let $\chi\in\mathcal{D}(\RR)$ with $\supp \chi\subset(-2\epsilon,2\epsilon)$ and $\chi(t)=1$ for $|t|\le \epsilon$. Let $\chi_+(t)=1-\chi(t)$ for $t>0$ and 0 elsewhere and $\chi_-=1-\chi-\chi_+$.  The functions $\chi$, $\chi_+$ and $\chi_-$ can be considered as functions on the Minkowski space $M=\RR^4$ which depend only on the time component. They form a smooth partition of unity subordinate to the cover 
\begin{figure}[htb]
  \begin{center}
    \def\svgwidth{0.75\textwidth}
    \input{drawing-mod.tex}
    \caption{The cover of $M$ (projected onto the $x^0-x^1$ plane) and the partition of unity $1=\chi_++\chi_-+\chi$. In addition the dependence region $C(f)$ of $S_{g\chi}(f)$ is drawn as the shaded region.}\label{fig:setting}
  \end{center}
\end{figure}
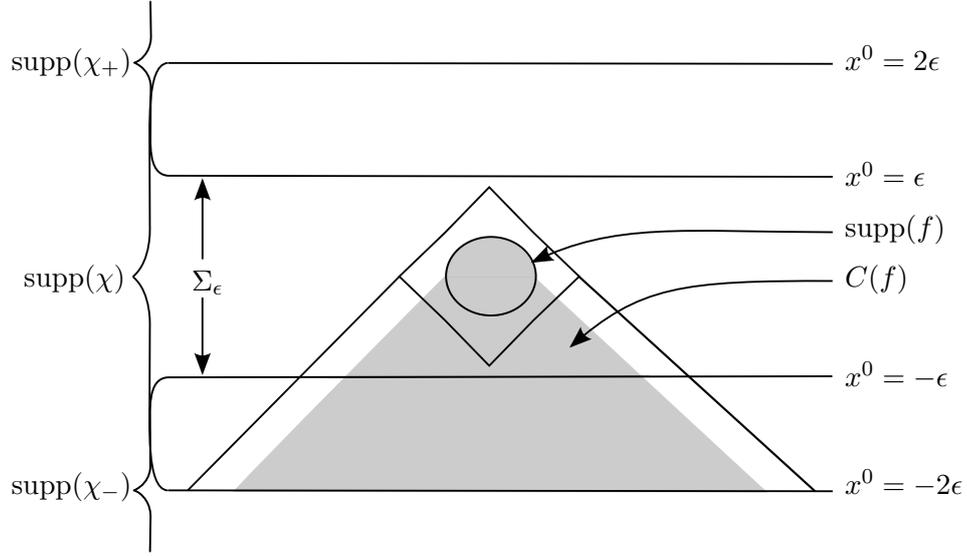
\begin{align*}
  M&= M_- \cup \Sigma_{2\epsilon} \cup M_+= I_- \times \RR^3 \cup  I_{2\epsilon} \times \RR^3 \cup I_+ \times \RR^3\ ,\\I_-&=(-\infty,-\epsilon) ,\; I_{2\epsilon}=(-2\epsilon,2\epsilon), \,I_+=(\epsilon,+\infty),\qquad \epsilon>0,
\end{align*}
see figure \ref{fig:setting}. We now decompose $g\in[\g]_{\mathcal O}$, 
\begin{equation*}
  g=g\chi_+ +g\chi+g\chi_- \ ,
\end{equation*}
use the causal factorization and obtain for $\supp f\subset \mathcal O\subset \Sigma_{\epsilon}$
\begin{align}
  S_g(f)\stackrel{\eqref{eq:causalfactor2}}{=}S_{g\chi+g\chi_-}(f)\stackrel{\eqref{eq:causalfactor3}}{=}S_{g\chi}(g\chi_-)^{-1}S_{g\chi}(f)S_{g\chi}(g\chi_-).\label{eq:causalpartition}
\end{align}
By construction, $S_{g\chi}(f)\in\mathfrak A(\mathcal O_1)$ for every neighborhood $\mathcal O_1$ of the dependence region $C(f)=J^-(\supp(f))\cap \supp(\chi)$, see figure \ref{fig:setting}.
The crucial fact is now that $S_{g\chi}(f)$ becomes independent of the choice of $g\in[\g]_{\mathcal O}$ for $\mathcal O$ sufficiently large. Thus \eqref{eq:causalpartition} induces an
embedding of the algebra of the interacting theory into that of the free theory,
\begin{align}
  \gamma_{\chi}: \mathfrak A_{[\g]} \to \mathfrak A,\quad S_{[\g]}(f)\mapsto S_{\g\chi}(f)\ ,\ \supp f\subset \Sigma_{\epsilon}\ ,\label{eq:interactingrepresent}
\end{align}
with
\begin{equation}
     S_{\g\chi}(f)=S_{g\chi}(f), \ g\in[\g]_{\mathcal O},\ \mathcal O\text{ sufficiently large}\ .
\end{equation}
As shown in \cite{cf}, $\gamma_{\chi}$ is even an isomorphism, and it maps local subalgebras into local subalgebras such that
\begin{equation}
  \gamma_{\chi}(\mathfrak A_{[\g]}(\mathcal O_r))\subset \mathfrak A(\mathcal O_{r+4\epsilon})\subset\gamma_{\chi}(\mathfrak A_{[\g]}(\mathcal O_{r+8\epsilon}))\ ,
\end{equation}
with $\mathcal O_r=\{(t,\vec x)||t|+|\vec x|\le r\}$.
This is in agreement with the Hamiltonian picture where the time zero algebras are identified.

We now restrict ourselves to time independent functions $\g$ and investigate the relation between the interacting dynamics and the free one. The free dynamics $\alpha_t$ acts on the relative $S$-matrices by shifting the test functions:
\begin{equation}
  \alpha_t(S_g(f))=S_{g_t}(f_t)\ , \ f_t(x^0,\mathbf{x})=f(x^0-t,\mathbf{x}) \ .\label{eq:timetranslation}
\end{equation}
The dynamics of the interacting theory $\alpha_t^{\g}$ is defined by the action on the generators $S_{[\g]}(f)$,
\begin{equation}
  \alpha^{\g}_t(S_{[\g]}(f)):= S_{[\g]}(f_t).\label{eq:timeev-full}
\end{equation}
Pointwise, for $g\in[\g]_{\mathcal O}$, where $\mathcal O\supset\supp(f)\cup\supp(f_t)$, this means
\begin{equation}
  \alpha^{\g}_t(S_{[\g]}(f))(g)= S_g(f_t)=\alpha_{t} (S_{g_{-t}}(f)).\label{eq:timeev-full1}
\end{equation}
We now use the isomorphism $\gamma_{\chi}$ and map the interacting dynamics to the algebra of the free theory,
\begin{equation}
   \alpha_t^{\g,\chi}\circ\gamma_{\chi}=\gamma_{\chi}\circ\alpha_t^{\g}\ .
\end{equation} 
For $t$ sufficiently small such that $\supp f_t\subset\Sigma_{\epsilon}$, we find
\begin{equation}
  \alpha^{\g,\chi}_t(S_{\g\chi}(f) )=S_{\g\chi}(f_t)=\alpha_t(S_{\g\chi_{-t}}(f)).\label{eq:timeev-timesl}
\end{equation}
The interacting dynamics $\alpha_t^{\g,\chi}$ and the free dynamics $\alpha_t$ differ by a co-cycle. To compute it, we note that, for sufficiently small $t$, the difference 
\begin{align}
  \chi_t-\chi=\pdiff{t} + \fdiff{t} \label{eq:difference}
\end{align}
can be written in terms of smooth functions $\pdiff{t},\fdiff{t}$ supported in the complement of the future and the past of $\supp(f)$, respectively, i.e.\ $\supp(\pdiff{t})\cap J^{+}(\supp(f))=\varnothing$ and vice versa.

Using the support property of  $\pdiff{t}$ and $\fdiff{t}$, a causal factorization can be done in complete analogy to \eqref{eq:causalpartition}.
\begin{theorem}\label{prop:cocycle}
  Let $\g$ have compact support in the spatial variable $\vec x$.
 Then the interacting dynamics $\alpha_t^{\g,\chi}$ and the free dynamics $\alpha_t$ are intertwined by a unitary co-cycle $t\mapsto U_{\g}^\chi(t)\in \mathfrak A$,
 \begin{align}
   \alpha_{t}^{\g,\chi}(A)&=U_{\g}^\chi(t)\alpha_t (A) U_{\g}^\chi(t)^{-1}, \ A\in\mathfrak A \ ,\label{eq:defcocycle}
  \end{align}
 with
  \begin{align}
	 U_{\g}^\chi(t)= S_{\g \chi}(\g \pdiff{t}).\label{eq:cocycle}
  \end{align}
  for $t$ sufficiently small, where $\pdiff{t}$ was defined in \eqref{eq:difference}.
\end{theorem}
The proof of equation \eqref{eq:defcocycle}  for sufficiently small $t$ is a straightforward application of the causal relations of the relative S-matrices (\ref{eq:causalfactor2}--\ref{eq:causalfactor3}):
\begin{align*}
  \alpha_t(S_{\g \chi}(f))= 	S_{\g \chi_t}(f_t)\stackrel{\eqref{eq:causalfactor2}}{=}S_{\g (\chi+\pdiff{t})}(f_t)\stackrel{\eqref{eq:causalfactor3}}{=}S_{\g \chi}(\g\pdiff{t})^{-1} S_{\g \chi}(f_t)S_{\g \chi}(\g \pdiff{t})\ .
\end{align*}
under the condition   $\supp(f),\,\supp(f_t)\subset \Sigma_{\epsilon}$. 
The condition of compact spatial support of $\g$
was necessary for the co-cycle $U_{\g}^\chi(t)$ to exist. 
In order to extend the formula to all values of $t$ we show that equation \eqref{eq:cocycle} actually characterizes a co-cycle.
\begin{proposition}
  For $t,s$ in a neighborhood of the origin, $t\mapsto U_{\g}^\chi(t)$ satisfies the co-cycle condition
  \begin{equation}
    U_{\g}^\chi(t+s)=U_\g^\chi(t)\alpha_t\left(U_\g^\chi(s)\right)\ .\label{eq:cocyclerelation}
  \end{equation}
\end{proposition}
\begin{proof}
  Let $\Theta^-=\theta(-t)$, where $\theta$ is the Heaviside function on $\RR$. Then $\pdiff{t}=\Theta^-(\chi_t-\chi)$, and for $t,s$ sufficiently small we have 
  \begin{align*}
    \pdiff{t+s}=\Theta^-(\chi_{t+s}-\chi)=\Theta^-(\chi_{t}-\chi)+\Theta^-(\chi_{t+s}-\chi_{t})=\pdiff{t}+(\pdiff{s})_t,
  \end{align*}
  where $(\pdiff{s})_t$ is $\pdiff{s}$ which is shifted by $t$, according to \eqref{eq:timetranslation}. We then find 
  \begin{align*}
    U_\g^\chi(t)^{-1}U_\g^\chi(t+s)=S_{\g \chi}(\g \pdiff{t})^{-1}S_{\g \chi}(\g \pdiff{t}+\g (\pdiff{s})_t))=S_{\g (\chi+\pdiff{t})}(\g (\pdiff{s})_t).
  \end{align*}
  Using $\chi+\pdiff{t}=\chi_t-\fdiff{t}$ and the fact that the supports of $\fdiff{t}$ and $(\pdiff{s})_t$ are causally separated, we can again apply the factorization rule and obtain for the right hand side
  \begin{equation*}
    S_{\g (\chi_t+\fdiff{t})}(\g (\pdiff{s})_t)=S_{\g \chi_t}(\g (\pdiff{s})_t)=\alpha_{t} \left(S_{\g \chi}(\g \pdiff{s})\right)=\alpha_t(U_\g^\chi(s))\ .
  \end{equation*}
  This shows the claim.
\end{proof}
The unitary $U_\g^\chi(t)$ corresponds to the time-evolution operator 
 \begin{align}
 	\textup e^{i(H_0+H_I) t} \textup e^{-iH_0t}\label{eq:cocylcetime0}
 \end{align}
 in the interaction picture.

We can now use the co-cycle condition to define  $U_\g^\chi(t)$ for all $t\in\RR$. The proposition guarantees that there exists a unique  extension. Thus the adjoint action of $U_\g^\chi(t)$
coincides  with the dynamics $\alpha_t^{G,\chi}\circ\alpha_{-t}$ of the interaction picture. This finishes the proof of the theorem.

The co-cycle satisfies the differential equation
 \begin{align*}
 	\frac{1}{i}\frac{d}{dt} U_\g^\chi(t)=\frac{1}{i}\frac{d}{ds}\Big|_{s=0}U_\g^\chi(t+s)= U_\g^\chi(t)\alpha_t(K^\chi_g)\ ,
 \end{align*}
with 
 \begin{align*}
 	K_\g^\chi:=\frac{1}{i}\frac{d}{dt}\Big|_{t=0} U_\g^\chi(t)&=[A(\g(-\frac{d}{dt}\pdiff{t}))]_{\g\chi}\ ,
 \end{align*}
 where the last equation follows from \eqref{eq:cocycle} and Bogoliubov's formula \eqref{eq:bogoliubov}.
 For the derivative of $\pdiff{t}$ one finds the following:
 \begin{align*}
   \pdiff{t}(x^0)=\Theta^-(x^0)\left(\chi_t(x^0)-\chi(x^0)\right),\qquad -\frac{d}{dt}\Big|_{t=0}\pdiff{t}(x^0)=\Theta^-(x^0)\dot\chi(x^0)=:\dot\chi^-(x^0)
 \end{align*}
where $\dot \chi$ is the first derivative of $\chi$. From the properties of $\chi$ it follows that $\supp(\dot\chi^-)\subset (-2\epsilon,-\epsilon)$ and that $\int dx^0\; \dot\chi^-(x^0)=1$. These  properties  imply that a smearing with $\dot\chi$ is in fact a time average  over the interval, in which the interaction is switched on, see figure \ref{fig:cutoff}.
\begin{figure}[tb]
  \begin{center}
    \def\svgwidth{0.5\textwidth}
    \input{cutoff.tex}
    \caption{The cutoff function $\chi$ as a dashed line. The first derivative of $\chi$ for $t<0$, i.e.\ $\dot{\chi}^-$ as solid line.}\label{fig:cutoff}
  \end{center}
\end{figure}
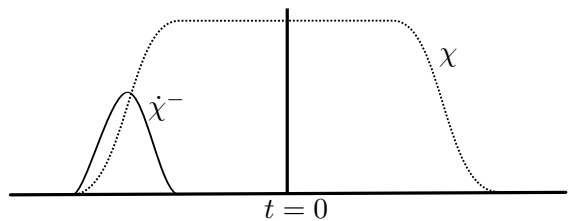
 
 We specialize now to the case $G=(h,0,\dots,0)$, $h\in\mathcal D(\RR^3)$,  and obtain
 \begin{align}
 	K_\g^\chi\equiv K_h^{\chi}=\left[\mathcal H_I(h\dot\chi^-)\right]_{h\chi}=\int d^4x\; h(\vec x) \dot \chi^-(x^0)\left[\mathcal H_I(x)\right]_{h\chi}\ .\label{eq:generator}
 \end{align}
We see that $K_h^{\chi}$ is the spatial integral of the time-averaged interaction Hamiltonian density, subject to the interaction $\mathcal H_I(h\chi)$. (By abuse of notation, we write $h\chi$ instead of $(h\chi,0,\dots,0)$.)

This corresponds to the situation in quantum mechanics where an instantaneous switching of the interaction leads to the generator
 \begin{align*}
   \frac{1}{i}\frac{d}{dt}\Big|_{t=0} \text e^{i(H_0+H_I)t}\text e^{-iH_0t}=H_I\ .
 \end{align*}
The differential equation for $U^\chi_h(t)$ has, for $t>0$, the solution
 \begin{align}
 	U_h^\chi(t)&=\mathbbm 1+\sum_{n=1}^\infty i^n \int_{0}^tdt_1\int_0^{t_1}dt_2\cdots \int_0^{t_{n-1}}dt_n\;\alpha_{t_n}(K_h^\chi)\cdots \alpha_{t_1}(K_h^\chi)\label{eq:dyson}
 \end{align}
in the sense of formal power series in $K_h^\chi$. 
Note that $K_h^\chi$ itself is a formal power series in the interaction with vanishing zeroth order term. Therefore \eqref{eq:dyson} is a well defined composition of formal power series.

A final result shows that the co-cycle $U_h^\chi(t)$ actually does, up to equivalence, not depend on the choice of the time cutoff $\chi$.
\begin{proposition}\label{prop:chi}
	Let $\chi,\chi'\in\mathcal D(\RR)$  such that $\supp(\chi),~\supp(\chi')\subset \Sigma_{2\epsilon}$ and $\chi(t)=1=\chi'(t)$ for $t\in(-\epsilon,\epsilon)$. It holds that the co-cycles $U_h^{\chi}$ and $U_h^{\chi'}$ are equivalent, i.e.
	\begin{align*}
		U_h^{\chi'}(t)=V^{-1} U_{h}^{\chi}(t) \alpha_t(V)
	\end{align*}
        with the unitary $V=S_{h\chi}(h\sigma^-)$, $\sigma^-=\Theta^-(\chi'-\chi)$.
\end{proposition}

\begin{proof}
  We set
  \begin{align*}
	\chi'-\chi&=\sigma^++\sigma^-\\
	\chi_t-\chi&=\fdiff{t}+\pdiff{t}\\
	\chi'_t-\chi'&=\fdifft{t}+\pdifft{t}
  \end{align*}
  where the index $\pm$ indicates that the supports of the functions are in the future or past of $\Sigma_\epsilon$, respectively. In particular, we find the relation $\sigma^-_t=\sigma^-+\pdifft{t}-\pdiff{t}$. 
  We then find for the right hand side of the asserted equation
    \begin{align*}
	&V^{-1} U_{h}^{\chi}(t) \alpha_t(V)=	S_{h\chi}(h\sigma^-)^{-1} S_{h\chi}(h\pdiff{t}) S_{h\chi_{t}}(h\sigma^-_t)\\
	&=S(h(\chi+\sigma^-))^{-1} S(h(\chi+\pdiff{t}))S(h\chi_{t})^{-1}S(h(\chi_{t}+\sigma_t^-))\\
	&=S(h(\chi'-\sigma^+))^{-1} S(h(\chi'-\sigma^+-\sigma^-+\pdiff{t}))S(h(\chi'_{t}-\sigma^+_t-\sigma^-_t))^{-1}S(h(\chi'_{t}-\sigma_t^+))\\
	&\stackrel{\eqref{eq:causalfactor0}}{=}S(h\chi')^{-1}S(h(\chi'-\sigma^-+\pdiff{t}))S(h(\chi'_t-\sigma^-_t))^{-1}S(h\chi_t')
     \end{align*}
where in the last line the causal factorization from equation \eqref{eq:causalfactor0} was used for the innermost two factors. The third factor is factorized, again with the help of \eqref{eq:causalfactor0}, and we obtain the equation
\begin{align*}
	S(h(\chi'_t-\sigma^-_t))^{-1}&=S(h(\chi'+\fdifft{t}+\pdifft{t}-\sigma^-_t))^{-1}\\&=S(h(\chi'+\pdifft{t}-\sigma^-_t))^{-1} S(h\chi') S(h(\chi'+\fdifft{t})^{-1})\\
&=S(h(\chi'+\pdiff{t}-\sigma^-))^{-1} S(h\chi') S(h(\chi'+\fdifft{t})^{-1})
\end{align*}
where we used the relation $\pdifft{t}-\sigma^-_t=\pdiff{t}-\sigma^-$ in the last line. Reinserting this factorization into the previous equation yields the desired result:
\begin{align*}
	&V^{-1} U_{h}^{\chi}(t) \alpha_t(V)=S(h(\chi'+\fdifft{t}))^{-1}S(h\chi'_t)\\
	&=S(h(\chi'+\fdifft{t}))^{-1}S(h(\chi'+\pdifft{t}+\fdifft{t}))\\
	&=S(h\chi')^{-1}S(h(\chi'+\pdifft{t}))=S_{h\chi'}(\pdifft{t})=U_{h}^{\chi'}(t)\ .
\end{align*}

\end{proof}
The generator $K_h^\chi$ changes under $\chi\to \chi'$ by
\begin{align*}
	K_h^{\chi'}=V^{-1}K_h^\chi V + V^{-1}\frac{1}{i}\frac{d}{dt}\Big|_{t=0} \alpha_t(V),\quad V=S_{h\chi}(h\sigma^-),
\end{align*}
as it is expected from a generator of a co-cycle.
In view of the last proposition we will often suppress  $\chi$ in the notation.

 The formalism developed here is very close to the Hamiltonian formalism used in quantum statistical mechanics, see \cite{brob}. 
A direct application is not possible due to the singularity of the interaction Hamiltonian density in relativistic field theory (see the work of St\"uckelberg \cite{stueckelberg}). By the transition from the time zero hypersurface to a time-slice with a finite extension in time a regularized interaction Hamiltonian density was found, which relates the free and the interacting dynamics. The price to pay is that the regularized density is a field in the interacting theory and thus only known as a formal power series in the interaction.
\section{KMS states for the interacting dynamics: General discussion}
After the construction of the interacting dynamics $\alpha_t^h$ with respect to the interaction Hamiltonian $\mathcal H_I$ with a spatial cutoff $h\in\mathcal D(\RR^3)$ we come to the main point of this work: The construction of KMS states for $\alpha_t^h$. 
First, in section \ref{sec:kmsfinite}, we show, that the perturbation technique, which was introduced by Araki \cite{arakikms} for bounded perturbations of $C^*$-dynamical systems, can be generalized to our framework. 
By this method we obtain KMS states for an interaction with a spatial cutoff.

Then, in section \ref{sec:alcondition}, we prove that the 
cutoff can be removed,
if the connected correlation functions of the unperturbed theory decay sufficiently fast in spatial directions.
\subsection{The case of finite volume}\label{sec:kmsfinite}
We start with a definition of a KMS state, adapted to our framework.
\begin{definition}\label{def:kms}
	Let $\alpha_t$ be a one-parameter group of automorphisms of a $^*$-algebra $\mathcal A$.	A state $\omega_\beta$ on $\mathcal A$ is called a KMS state with respect to $\alpha_t$ 
	with inverse temperature $\beta$, if  the functions
	\begin{align*}
	  (t_1\ldots,t_n)\mapsto\omega_\beta(\alpha_{t_1}(A_1)\cdots \alpha_{t_n}(A_n)),\qquad A_1,\ldots,A_n\in \mathcal A,
         \end{align*}
        have an analytic continuation to the region
        \begin{align*}
	 \mathfrak T_n^\beta=\{(z_1\ldots,z_n)\in\CC^n: 0<\Im(z_i)-\Im(z_j)<\beta,\quad 1\le i<j\le n\} \ ,
        \end{align*}
        are bounded and continuous on the boundary and fulfill the boundary conditions
       \begin{align*}
	 &\omega_\beta(\alpha_{t_1}(A_1)\cdots \alpha_{t_{k-1}}(A_{k-1})\alpha_{t_{k}+i\beta}(A_k)\cdots \alpha_{t_n+i\beta}(A_n))\\
	 &=\omega_\beta(\alpha_{t_k}(A_k)\cdots \alpha_{t_{n}}(A_{n})\alpha_{t_{1}}(A_1)\cdots \alpha_{t_{k-1}}(A_{k-1}))\ .
       \end{align*}
\end{definition}
In the $C^*$-algebraic framework (cf. \cite{haagbuch}) all higher order conditions are implied by the condition on the $n=2$ point function 
\begin{align*}
	t\mapsto \omega_\beta(A\alpha_t(B))\ . 
\end{align*}
This is due to the better analytic control in the realm of $C^*$-algebras compared to algebras of unbounded operators.

The interacting KMS state can be obtained using the standard methods which were first developed by Araki  \cite{arakikms} and are elaborated in more detail in \cite{brob}.  The non-trivial part of the proof is to show that the methods which were developed in the $C^*$-algebraic formalism extend to our framework. 

Let $h\in\mathcal D(\RR^3)$. We set $\g(t,\vec x)=(h(\vec x),0,\dots,0)$ and replace everywhere the index $\g$ by $h$. Moreover, we suppress the index $\chi$.
\begin{theorem}\label{thm:analyticcont}
  Let $\omega_\beta$ be a KMS state on $\mathfrak A$ with respect to $\alpha_t$.  Then the following statements hold in the sense of formal power series in the interaction:
  \begin{itemize}
	\item For $A_1,\ldots,A_n\in\mathfrak A_h$ 
		the functions
	         \begin{align*}
		  (t_1,\ldots,t_n,s)\mapsto \omega_\beta\left(\gamma\left(\alpha^{h}_{t_1}(A_1)\cdots \alpha_{t_n}^{h}(A_n)\right) U_h(s)\right)
	         \end{align*}
	         have an analytic continuation into $\mathfrak T_{n+1}^\beta$ and are bounded and continuous on the boundary.
	\item Let $\omega_{\beta}^{h}(A)$ denote, for every $A\in\mathfrak A_{[h]}$, the evaluation of the analytic extension of the function
                  \begin{align}
	            t\mapsto G_A(t)=\frac{\omega_\beta(\gamma(A) U_h(t))}{\omega_\beta(U_h(t))}\label{eq:perturbedkms}
                  \end{align}
                 at $i\beta$. Then $\omega^{h}_{\beta}$	
                 is a KMS state with respect to $\alpha_t^{h}$ on $\mathfrak A_{[h]}$. 
\end{itemize}
\end{theorem}
\begin{proof}
  We begin with the proof of the first item. To show that the analytic continuation of
  \begin{align*}
	G_{n}(t_1,\ldots,t_n|s)= \omega_\beta\left(\gamma\left(\alpha^{h}_{t_1}(A_1)\cdots \alpha^{h}_{t_n}(A_n)\right) U_h(s)\right)
  \end{align*}
  is well-defined, it is useful to construct a unitary operator intertwining the dynamics at different times. To this end, consider
  \begin{align*}
	U_h(t,s)&=U_h(t)^{-1}U_h(s)\ .
  \end{align*}
  With the initial condition $U_h(t,t)=\mathbbm 1$ 
  one obtains the power series expansion
  \begin{align*}
	U_h(t,s)=\sum_{n=0}^\infty (i(t-s))^n \int_{\mathcal S_n}d^nu \;\alpha_{s+u_1(t-s)}(K_h)\cdots \alpha_{s+u_n(t-s)}(K_h)\ ,
  \end{align*}
  where $\mathcal S_n$ denotes the unit simplex
  \begin{align}
	\mathcal S_n=\{(u_1,\ldots,u_n)\in \RR^n: 0\le u_1\le \ldots \le u_n\le 1\}\ .\label{unitsimplex}
  \end{align}
  We use the representation of $\alpha_t^{h}$,  
  \begin{align*}
	\gamma\circ \alpha_t^{h}=\mathrm{Ad}(U_h(t))\circ \alpha_t\circ\gamma\ ,
  \end{align*}
  and insert the expansion of $U_h$  into $G_n$.  We find for the $l$-th term $G_n^{(l)}$ in the formal power series expansion of $G_{n}$ in $K_h$ 
  \begin{align*}
	G_n^{(l)}&(t_1,\ldots,t_n|s)\\
	=&\sum_{\substack{\vec m\in \NN_0^{n+1}\\\abs{\vec m}=l}} \int_{\mathcal S_{m_1}} d^{m_1}u^{(1)}\cdots  \int_{\mathcal S_{m_{n+1}}}d^{m_{n+1}}u^{(n+1)} (it_1)^{m_1}\cdots (i(t_n-s))^{m_{n+1}}\times \\
	& \omega_\beta\left(\prod_{j=1}^{m_1}\left[\alpha_{u^{(1)}_{j}t_1}(K_h) \right]\alpha_{t_1}(A_1)\prod_{j=1}^{m_2} \left[\alpha_{t_2+u_{j}^{(2)}(t_2-t_1)}(K_h)\right]\cdots \prod_{j=1}^{m_{n+1}} 
	\left[\alpha_{s+u_{j}^{(n+1)}(t_n-s)}   (K_h)\right] \right).
  \end{align*}
  Using the restriction on the integration variables and the fact that $\omega_\beta$ is a KMS state it follows that the functions $G_n^{(l)}$ 
  have an analytic continuation into the strip $\mathfrak T_{n+1}^\beta$. Moreover, if we fix $s=i\beta$ we find that the functions
  \begin{align*}
	(t_1,\ldots,t_n)\mapsto G_n(t_1,\ldots,t_n|i\beta)
  \end{align*}
  can be analytically extended to 
  \begin{align*}
	\{(z_1,\ldots,z_n)\in\CC^n: 0<\Im(z_1)<\Im(z_2)<\ldots<\Im(z_{n})<\beta\}
  \end{align*}
  and fulfill the KMS boundary conditions in the sense of formal power series in $K_h$ (thus in $\mathcal H_I$ by composition of formal power series) as in Definition \ref{def:kms}. Thus the linear functional 
  $A\mapsto G_A(i\beta)$   defines a KMS functional. 

  In order to prove the second point it remains to show that the functional is positive. 
  This follows by noting that, in the GNS representation $\pi$ induced by $\omega_{\beta}$ with cyclic vector $\Omega$, the vector valued function
  \begin{equation}
     t\mapsto \Psi(t)=U_h(t)\Omega
  \end{equation} 
  has an analytic extension to the strip $0<\Im(t)<\frac{\beta}{2}$, and that $\Psi(i\frac{\beta}{2})$ induces $\omega_{\beta}^{h}$,
  \[\omega_{\beta}^h(A)=\frac{\langle \Psi(\frac{i\beta}{2}),\pi(A)\Psi(i\frac{\beta}{2})\rangle}{\langle \Psi(i\frac{\beta}{2}),\Psi(i\frac{\beta}{2})\rangle}\ .\]
  This is, in the $C^*$-algebraic setting, a well known consequence of the KMS condition 
  for $\omega_{\beta}$   and the co-cycle relation for $U_h$ (see, e.g. \cite{brob}), and the argument holds as well in our framework.
  A complication is that one needs an appropriate  concept of positivity for formal power series. Here we call a formal power series of complex numbers positive if it can be written as an absolute square of a power series.
  See \cite{waldmann,lindner} for a more detailed discussion on states on algebras of formal power series .
\end{proof}

There exists a convenient expansion of the interacting KMS state in terms of the connected correlation functions  of the free theory that will play an important role in the discussion on the spatial cluster properties. In the rest of the work we use the implicit definition
\begin{align*}
  \omega(A_1\cdots A_n)=\sum_{P\in \text{Part}\{1,\ldots,n\}} \prod_{I\in P} \omega^c\left(\bigotimes_{i\in I} A_i\right)
\end{align*}
for the connected part\footnote{The notion ``truncated part'' of the state is sometimes used in the literature. Both are synonyms.} $\omega^\text{c}$ of $\omega$, 
seen as a linear functional $\omega^\text{c}: \mathfrak T\mathfrak A\to \CC$, where $\mathfrak T\mathfrak A$ is the tensor algebra over $\mathfrak A$ and $\text{Part}\{1,\ldots,n\}$ is 
the set of all partitions of $\{1,\ldots,n\}$ into non-void subsets.  Note that the properties of $\omega_\beta$ from Definition  \ref{def:kms} carry over to the connected part $\omega_\beta^\text{c}$. 
The following proposition was first proven in \cite{brattrunc} in the $C^*$-algebraic setting.
\begin{proposition}\label{prop:connected}
The KMS state $\omega_\beta^{h}$ can be written in terms of the analytically extended connected  correlation functions:
	\begin{align}
		\omega_{\beta}^{h}(A)=\sum_{n=0}^\infty(-1)^n\int_{\beta\mathcal S_n} d^nu\,  
		 \omega^c_\beta\left(\gamma(A)\otimes \alpha_{iu_1}(K_h)\otimes \dots\otimes \alpha_{iu_n}(K_h)\right)\label{eq:kmsexpansion}
	\end{align}	
\end{proposition}
\noindent(Note that there are no analytic elements in the algebra $\mathfrak A$. The suggestive notation above is always to be understood in the sense of analytic extensions of correlation functions.) 

\begin{proof}
The proof proceeds in the same way as
the original proof in \cite{brattrunc} and is only sketched here. For this  we introduce the following expansion in the interaction $K_h$ with a formal parameter $\lambda$,
\begin{align*}
	\omega_\beta^{h,\lambda}(A)&=\sum_{n=0}^\infty \lambda^n\Omega_n(A)\ ,\qquad  \omega_\beta(\gamma(A)U_{h,\lambda}(i\beta))=\sum_{n=0}^\infty \lambda^n\nu_n(A) \ .
\end{align*}
The coefficients $\nu_n$ are obtained from the expansion of  $U_{h,\lambda}(t)$:
\begin{align*}
 \nu_n(A)=(-\beta)^n\int_{\mathcal S_n} du_1\cdots du_n \; \omega_\beta(\gamma(A) \alpha_{iu_1\beta}(K_h)\cdots \alpha_{iu_n\beta}(K_h)),\quad \nu_0(A)=\omega_\beta(\gamma(A))\ .
\end{align*}
By definition of the interacting states $\omega_\beta^{h,\lambda}$ it holds that
\begin{align*}
&\omega_\beta(\gamma(A)U_{h,\lambda}(i\beta))=\omega_\beta(U_{h,\lambda}(i\beta))\omega_\beta^{h,\lambda}(A)\ .
\end{align*}
Thus, by comparing the coefficients of the expansions on both sides, one gets
\begin{align*}
	\nu_n(A)&=\sum_{k=0}^n \nu_{k}(\mathbbm 1) \Omega_{n-k}(A), \qquad  	\Omega_n(A)=\nu_n(A)-\sum_{k=1}^n \nu_k(\mathbbm 1) \Omega_{n-k}	(A)\ .
\end{align*}
By induction, it is then shown that 
\begin{align*}
	\Omega_n(A)=(-\beta)^n \int_{\mathcal S_n} du_1\cdots du_n \;\omega_\beta^c\left(\gamma(A)\otimes \alpha_{iu_1\beta}(K_h) \otimes \cdots \otimes \alpha_{iu_n\beta}(K_h)\right)\ .
\end{align*}
\end{proof}

Having a closer look on the definition of the perturbed KMS state $\omega_\beta^{h}$ in Proposition \ref{thm:analyticcont}
one may ask, whether it depends on the choice of $\chi$. Taking into acount the dependence of $\gamma$ and $U_h$ on $\chi$ we find
that the state is, in fact, independent of the choice of $\chi$:
\begin{proposition}\label{prop:independent}
	Let $\chi,\chi'\in\mathcal D(\RR)$ be given as in Proposition \ref{prop:chi}. 
	Then the associated KMS states on $\mathfrak A_h$, as defined in Theorem \ref{thm:analyticcont} coincide.
\end{proposition}
\begin{proof}
  In the same notation as in the proof of Proposition \ref{prop:chi} we denote $\chi'-\chi=\sigma_++\sigma_-$. Due to causality of the relative S-matrices (\ref{eq:causalfactor2}--\ref{eq:causalfactor3}) we know that
  \begin{align*}
		\gamma_{\chi'}=\mathrm{Ad}(V^{-1})\circ \gamma_{\chi} \ .
  \end{align*}
  Furthermore we use the transformation property of the co-cycle,   
  \begin{align*}
    U_h^{\chi'}(t)=V^{-1}U_h^{\chi}(t) \alpha_t(V),
  \end{align*}
  which is derived in Proposition \ref{prop:chi}. The state $\omega_{\beta}^{h\chi'}$ is defined, via the equation \eqref{eq:perturbedkms}, by analytic extension of the functions $G'_A(t)$ where $\chi$ is replaced by $\chi'$.
  Looking at the numerator we find that 
  \begin{align*}
	\omega_\beta(\gamma_{\chi'}(A)  U_{h}^{\chi'}(t))&= \omega_\beta\left(V^{-1}  \gamma_{\chi}(A)  U_{h}^\chi(t)  \alpha_{t}(V)\right)\ .
  \end{align*}
  Due to the KMS condition, the analytic extension to $t=i\beta$ coincides with that of the numerator of $G_A$. The same argument holds for the denominator by setting $A=1$.
\end{proof}
\subsection{The case of infinite volume}\label{sec:alcondition}
On the level of the algebra and the time evolution, the adiabatic limit $h\to 1$ is easy. 
In the following we assume that the cutoff function $h\in \mathcal D(\RR^3)$ is chosen such that $h(\vec x)=1$ for all $\vec x\in B_r$, 
where $B_r$ is the open ball with radius $r$ in $\RR^3$, and we denote the constant function with value 1 by $I$.
Let $\supp f\subset \mathcal O_r$. Then
\begin{equation}
  S_{[h]}(f)(g)=S_g(f)=S_{{[I]}}(f)(g)\ ,
\end{equation}
if $g\equiv 1$ on $\mathcal O_r$, hence $\mathfrak A_{[h]}(\mathcal O_r)=\mathfrak A_{[I]}(\mathcal O_r)$.

Moreover, if also $\supp f_t\subset \mathcal O_r$, then
\begin{equation}
\alpha^h_t(S_{[I]}(f))=\alpha^I_t(S_{[I]}(f)) \ ,
\end{equation} 
hence $\alpha^h_t(A)=\alpha^I_t(A)$ for $A\in\mathfrak A_{[I]}(\mathcal O_{r-|t|})$.

We now
would like to show that the expectation values $\omega_\beta^{h}(A)$ with $A\in\mathfrak A_{[I]}(\mathcal O)$,  which are determined by \eqref{eq:kmsexpansion} for sufficiently large $r$, converge as $h\to I$ in the sense of formal power series in 
$\mathcal H_I$. 

The adiabatic or thermodynamic limit in which $h$ tends to the constant function has to be approached in  such a way that the boundary volume
\begin{align*}
	\{\vec x\in\RR^3: 0<h(\vec x)<1\}
\end{align*}
becomes negligible as the volume of the region, in which $h=I$, grows to infinity (assuming that $h$ is positive). The precise formulation of this idea is due to van Hove \cite{vanhove}. (See the monograph of Ruelle \cite{ruelle} for more details.)

\begin{definition}\label{def:adlimit}
	A van Hove sequence of cutoff functions is a sequence $(h_n)_{n\in\NN}$ of test functions $h_n\in\mathcal D(\RR^3)$ with the following properties:
	\begin{align*}
		 0\le h_n(\vec x)\le 1 \quad \forall \vec x\in \RR^3,\quad h_n(\vec x)=\begin{cases} 1 &,\  \abs{\vec x}<n\ ,\\ 0 &,\ \abs{\vec x}>n+1\ .\end{cases}
	\end{align*}
	The thermodynamic limit in the sense of van Hove is the limit $\lim_{n\to\infty} h_n$ for all van Hove sequences. It is abbreviated by $\mathrm{vH-}\lim_{h\to I}$.
\end{definition}

In the expansion of the interacting KMS state in connected correlation functions of the generator $K_h$, as defined in \eqref{eq:generator} ,
the spatial cutoff $h$ enters in two different ways; namely in the modification of the interaction density by the interaction $\mathcal H_I(h\chi)$ and in the spatial integral over the density with the  spatial cutoff.
We may modify $K_h$ by considering the interaction density in the adiabatic limit $h\to I$ and obtain another generator
\begin{align}
  K_h'	:=	[\mathcal H_I({h\dot\chi^-})]_{\chi}\label{eq:generatoral}
\end{align}
instead of $K_h$. 
One readily sees that the difference of $K_h$ and $K_h'$ is localized in a small neighborhood of that subregion of the support of $h$ where $h\neq I$, intersected with the time-slice $\Sigma_{2\epsilon}$.
Consequently, the dynamics induced by $K_h$ and $K_h'$ will coincide in the adiabatic limit. 

The latter generator has the advantage that it is a \emph{linear functional} of $h$,
\begin{equation}
  K'_h=\int d^3\vec x\, h(\vec x)\alpha_{0,\vec x}(R),
\end{equation}
with
\begin{align}
	R=\int dt\;\dot\chi^-(t) [\mathcal H_I(t,0)]_{\chi}\ .\label{eq:generatoradlim}
\end{align}
Here $\alpha_{t,\vec x}$ denotes the automorphism of $\mathfrak A$ representing a  translation in Minkowski space. 
$R$ is independent of $h$. Therefore we only have to control the decay behavior of the connected correlation functions in order to prove the existence of the adiabatic limit. 
\begin{theorem}\label{thm:adlimcond}
	Let $\omega'_\beta{}^{h}$ be the interacting KMS state for the time evolution $\alpha'_{t}{}^{h}$ where the generator $K_h$ is replaced by $K'_h$. 
	If the (analytically extended) connected correlation functions 
        \begin{align*}
          F_{n}(u_1,\vec x_1;\cdots ;u_n,\vec x_n)=\omega^\textup{c}_\beta\left(A_0\otimes \alpha_{iu_1,\vec x_1}(A_1)\otimes \cdots \otimes  \alpha_{iu_n,\vec x_n}(A_n)\right)
       \end{align*}
       are contained in the space $L^1(\beta\mathcal S_n\times \RR^{3n})$ for  $A_i\in\mathfrak A$ with $i=0,\ldots,n$, $n\in\NN$ and $0<\beta\le +\infty$, then the van Hove limit 
       \begin{align*}
	\mathrm{vH-}\lim_{h\to1}\omega'_\beta{}^{h}(A)=	
	:\omega_\beta^{I}(A), \qquad A\in \mathfrak A_{[I]},
       \end{align*}
       exists and defines a KMS state on $\mathfrak A_{[I]}$  with respect to $\alpha_t^I$.
\end{theorem}
\begin{proof}
  Due to Proposition \ref{prop:connected} we know that the expectation value of $A\in\mathfrak A$ in the interacting state $\omega'_\beta{}^{h}\circ \gamma^{-1}$ 
  can be written in terms of the connected correlation functions. Inserting the definition of $K'_h$ we obtain terms of the form
  \begin{align*}
    \int_{\beta\mathcal S_n} d^nu\int d^{3n}\vec x\, h(\vec x_1)\dots h(\vec x_n)\omega_\beta^{\textup c}( A\otimes \alpha_{iu_1,\vec x_1}(R)\otimes\dots\otimes \alpha_{iu_n,\vec x_n}(R)) \ .
   \end{align*}
   But $R$ is a formal power series with coefficients in $\mathfrak A$. Hence each term in the expansion is of the form of an integral over a correlation function of the form $F_n$ as defined in the theorem. 
   Hence under the assumption   on $F_n$ the limit $h\to I$ exists.
\end{proof}
It is clear from the proof that in order to establish the convergence of the state as $h$ tends to $I$, the strict assumptions on the 
sequence $(h_n)_{n\in\NN}$ can be weakend. In fact, any bounded sequence of functions approaching the constant 1 uniformly on compact sets will do. 
However, if one is interested in observables like the free energy per volume, assumptions on differentiability and the control on the ``boundary'', i.e. the region in which $h_n$ drops from the constant to zero have to be made. See \cite{lindner} for more details on this.

A natural question that emerges at this point is, whether a corresponding theorem holds for the interacting KMS states defined by the co-cycle $U_h^\chi$ by formula \eqref{eq:perturbedkms} in the previous section. 
In this case the interaction density at points near to the boundary of the support of $h$ is not yet in the adiabatic limit. Let   $ [A]_{g}^{(m)}$ denote the $m$-th term in the 
power series expansion of $[A]_{g}$, let $g_0,\dots g_k\in\mathcal{D}(\RR^4)$ and 
$\vec m\in\NN_0^{k}$. We consider the functions
\begin{align*}
   F_k^{\vec m}(\vec x_1,\ldots,\vec x_k)(g_0,\dots,g_k)=	\omega_\beta^{\text{c}}\left([A_0]_{g_0}^{(m_0)}\otimes\alpha_{\vec x_1}[A_1]_{g_1}^{(m_1)}\otimes \cdots \otimes \alpha_{\vec x_k}[A_k]_{g_k}^{(m_k)}\right) \ .
\end{align*}
If, for all $\vec m$, these functions are uniformly bounded by an $L^1$ function, for $g_0,\ldots, g_k$ in a bounded set of $\mathcal D(\RR^4)$, the 
proof of the theorem applies, provided the derivatives of the functions $h_n$ are uniformly bounded in $n$.

We expect that also this property holds for the case of a massive free field, but restrict ourselves to the simpler case treated in the theorem. Unfortunately, a direct proof of the $\chi$-independence of our construction as in Proposition \ref{prop:independent}
is not yet available for this case. It would be an immediate consequence of the stronger estimate described above.
\section{Cluster properties of the massive scalar field}
In this section we will show that there exist examples in which the conditions of Theorem~\ref{thm:adlimcond} are fulfilled. To this end, we start from the algebra $\mathfrak A$ of Wick polynomials of the free, massive scalar field. It is well-known that the quasi-free states $\omega_\text{vac}$ and $\omega_\beta$ induced by the 2-point functions 
\begin{align*}
	\omega_{\text{vac}}(\Phi(x)\Phi(y))&=D_+^\text{vac}(x-y),\; D_+^\text{vac}(x)=\frac{1}{(2\pi)^3}\int \frac{d^3\vec p}{2\omega_{p}} \text e^{-i(\omega_{p}x^0-\vec p \vec x)}\\
	\omega_{\beta}(\Phi(x)\Phi(y))&=D_+^\beta(x-y), \; D_+^\beta(x)=\frac{1}{(2\pi)^3}\int \frac{d^3\vec p \;\text e^{i\vec p \vec x}}{2\omega_{p}(1-\text e^{-\beta \omega_p})}   
	\left(\text e^{-i\omega_{p}x^0} +\text e^{i\omega_p (x^0+i\beta)} \right)
\end{align*}
where $\omega_p=\sqrt{\vec p^2+m^2}$ with  $m^2\ge 0$, define KMS states on the Weyl algebra of the free field. In Appendix \ref{sec:kms} we show that the extension of these states to the algebra $\mathfrak A$ have the desired analytic properties
which are required in  Definition \ref{def:kms}. 

We will show that these states satisfy the conditions in Theorem \ref{thm:adlimcond}, implying the existence of the adiabatic limit of their interacting counterparts with respect to the time-evolution $\alpha_t^I$ for polynomial interactions $\mathcal H_I$. In particular the interacting state $\omega_\text{vac}^I$ is shown to define a translation-invariant ground state.  This provides a new construction of the vacuum state, independent of the known construction of  Epstein and Glaser  \cite{eg1}.

Subsequently, we derive the existence of thermal equilibrium states of the interacting, massive scalar field. To the best of our knowledge, this is the first complete proof within renormalized perturbation theory in QFT.
\subsection{Vacuum state}
The main point of this section is a theorem that shows that the connected vacuum $n$-point correlation functions of (composite) free fields decrease 
exponentially in spatial and imaginary time directions.  This, in turn, implies the existence of the adiabatic limit of the interacting vacuum state (in the sense of van Hove) due to 
Theorem \ref{thm:adlimcond}. Such a statement has already been proven in the lecture notes by Araki in \cite{arakilecture} in an axiomatic setting, where the $A_i$ are bounded 
operators and translations $\alpha_{t_i,\vec x_i}$ with uniformly bounded, real times $t_i$ were used.
\begin{proposition}\label{prop:clustervac}
  Let $\omega_{\textup{vac}}$ be the vacuum state of the free Klein-Gordon field with mass $m>0$, induced by the translation invariant two-point function
  \begin{align}
	\omega_\textup{vac}(\Phi(x)\Phi(0))= D_+^{\textup{vac}}(x)=\frac{1}{(2\pi)^3}\int \frac{d^3\vec p}{2\omega_{p}} \text e^{-i(\omega_{p}x^0-\vec p \vec x)}\label{eq:vacuum2pf}\ .
  \end{align}
  The connected correlation function
  \begin{align*}
	F_{n}^{\textup{vac}}(u_1,\vec x_1; \cdots ;u_n,\vec x_n)=\omega_{\textup{vac}}^\textup{c}\left(A_0\otimes \alpha_{iu_1,\vec x_1}(A_1)\otimes \cdots \otimes \alpha_{iu_n,\vec x_n}(A_n)\right)
  \end{align*}
  for $A_0,\ldots,A_n\in\mathfrak A(\mathcal O)$ with $\mathcal O\subset B_R\subset \RR^4$ decreases exponentially in $\mathcal S_n^\infty\times \RR^{3n}$
  \begin{align*}
	\abs{F_n^{\textup{vac}}(u_1,\vec x_1;\ldots;u_n\vec x_n)}\le c \text e^{-\frac{m}{\sqrt{n}} r_e}, \qquad r_e=\sqrt{\sum_{i=1}^n u_i^2+\abs{\vec x_i}^2}\ ,
  \end{align*}
  uniformly for $r_e>2 R$. Here $\mathcal S_n^\infty=\{(u_1,\ldots,u_n)\in\RR^n: 0<u_1<\cdots < u_n\}$.
\end{proposition}

\begin{proof}
  We use the off-shell formalism explained in Appendix~\ref{sec:kms} by which the elements of $\mathfrak A$ can be identified with functionals on field configurations $\phi\in\mathcal C^{\infty}(M)$, 
  up to functionals which vanish on solutions of the Klein-Gordon equation. The connected correlation functions $\omega_{\textup{vac}}^c$ can be written in terms of the functional differential operator 
  $\Gamma_2^{ij}$ from the proof of 
  Proposition  \ref{lemma:kmsanalytic}, where the  KMS two-point function $D_+^\beta$ has to be replaced by $D_+^\textup{vac}$. The correlation function $\omega_{\textup{vac}}$ itself can be written as
  \begin{align*}
    \omega_{\textup{vac}}(A_0 A_1\cdots A_n)=	\prod_{0\le i<j \le n}\textup e^{\Gamma_2^{ij}} (A_0\otimes \cdots \otimes A_n)\Big|_{(\phi_0, \dots, \phi_n)=0}\ .
  \end{align*}
  Here the product of  exponentials can be rewritten as 
  \begin{align*}
	\prod_{0\le i<j \le n}\textup e^{\Gamma_2^{ij}}&=\prod_{0\le i<j\le n}\sum_{m=0}^\infty \frac{(\Gamma_2^{ij})^{m}}{m!}=\sum_{\substack{l=(l_{ij})_{i<j}\\l_{ij}\in\NN_0}} \prod_{i<j}\frac{(\Gamma_2^{ij})^{l_{ij}}}{l_{ij}!}
  \end{align*}
  which reads in terms of a graphical expansion
  \begin{align*}
	\sum_{\substack{l=(l_{ij})_{i<j}\\l_{ij}\in\NN_0}} \prod_{i<j}\frac{(\Gamma_2^{ij})^{l_{ij}}}{l_{ij}!}=\sum_{G\in\mathcal G_{n+1}}\Gamma_G,\qquad \Gamma_G=\prod_{i<j}\frac{(\Gamma_2^{ij})^{l_{ij}}}{l_{ij}!}\ ,
  \end{align*}
  where $\mathcal G_n$ denotes the set of all graphs $G$ with $n$ vertices and $l_{ij}$ are the number of lines joining the vertices $i$ and $j$. 
  Rewriting the products of exponentials in another way and using a similar argument as above one finds
  \begin{align}
    \prod_{0\le i<j \le n}\textup e^{\Gamma_2^{ij}}=	\prod_{0\le i<j \le n}\left(\textup e^{\Gamma_2^{ij}}-1+1\right)=\sum_{G\in\mathcal G_{n+1}}\prod_{G'\in [G]} 
    \prod_{i<j} \left(\frac{(\Gamma_2^{ij})^{l_{ij}}}{l_{ij}!}\right)\label{eq:onetrick}
  \end{align}
  where $[G]$ denotes the set of connected components of $G$. The connected correlation functions can be consequently written as 
  \begin{align*}
    \omega_{\textup{vac}}^\textup{c}(A_0\otimes \cdots \otimes A_n)=	\sum_{G\in\mathcal G_{n+1}^\textup{c}} \prod_{i<j} \left(\frac{(\Gamma_2^{ij})^{l_{ij}}}{l_{ij}!}\right) (A_0\otimes \cdots \otimes A_n)\Big|_{(\phi_0,\dots,\phi_n)=0}
  \end{align*}
  where $\mathcal G_{n}^\textup{c}$ denotes the set of connected graphs with $n$ vertices. The last equation can be verified by showing that the recursion formula for the connected correlation 
  functions picks out exactly the connected   components in the graphical expansion on the right hand side of equation~\eqref{eq:onetrick}.

  Then the functions $F^{\textup{vac}}_{n}$ can be written as 
  \begin{align*}
    F^{\textup{vac}}_{n}(u_1,\vec x_1; \cdots ;u_n,\vec x_n)&=\omega_{\textup{vac}}^\text{c}\left(A_0\otimes \alpha_{iu_1,\vec x_1}(A_1)\otimes \cdots \otimes \alpha_{iu_n,\vec x_n}(A_n)\right)\\&=
    \sum_{G\in\mathcal G_{n+1}^\textup{c}} \prod_{i<j} \left(\frac{(\Gamma_2^{ij})^{l_{ij}}}{l_{ij}!}\right) (A_0\otimes \cdots \otimes \alpha_{iu_n,\vec x_n}(A_n))\Big|_{(\phi_0,\dots,\phi_n)=0}\\
    &=: \sum_{G\in\mathcal G_{n+1}^\textup{c}} \frac{1}{\textup{Symm}(G)} F^{\textup{vac}}_{n,G}(u_1,\vec x_1; \cdots ;u_n,\vec x_n)
  \end{align*}
  with
  \begin{equation}
    F^{\textup{vac}}_{n,G}(u_1,\vec x_1; \cdots ;u_n,\vec x_n)=  \prod_{i<j} (\Gamma_2^{ij})^{l_{ij}} (A_0\otimes \alpha_{iu_1,\vec x_1}(A_n) \otimes \cdots \otimes \alpha_{iu_n,\vec x_n}(A_n))\Big|_{(\phi_0,\dots,\phi_n)=0}\ ,\nonumber
  \end{equation}
  similar to the terminology of the proof of Proposition \ref{lemma:kmsanalytic}. The source and range of the line $l$ is denoted with $s(l)$ and $r(l)$, respectively, and $X,Y$ contain all 
  points which are connected by the lines $l\in E(G)$. The last line defined the contribution of $G$ to $F_n^{\text{vac}}$ and $\text{Symm}(G)$ is the symmetry factor of $G$. 
  Instead of indexing all vertices we can also index the graph $G$ by all its edges $l\in E(G)$. The contribution of  a fixed, connected graph $G$ is
  \begin{align*}
    &F^{\textup{vac}}_{n,G}(u_1,\vec z_1; \cdots ;u_n,\vec z_n)\\
    &= \int  dX\;dY\;\prod_{\substack{l\in E(G)}} D_+^\textup{vac}(x_l-y_l) \frac{\delta^2}{\delta\phi_{s(l)}(x_l)\delta \phi_{r(l)}(y_l)} (A_0\otimes \cdots \otimes \alpha_{iu_n,\vec z_n}(A_n))\Big|_{(\phi_0,\dots,\phi_n)=0}\\
    &=\int dX\;dY\; \prod_{\substack{l\in E(G)}} D_+^\textup{vac}(\bar x_l-\bar y_l) \; \Psi(X,Y)
  \end{align*}
  with the abbreviations $\bar x_l= (x^0_l+iu_{s(l)},\vec x_l+\vec z_{s(l)})$ and $\bar y_l= (y^0_l+iu_{r(l)},\vec y_l+\vec z_{r(l)})$ and the functional derivatives
  \begin{align*}
	\Psi(X,Y)&=\prod_{l\in E(G)}  \frac{\delta^2}{\delta \phi_{s(l)}(x_l) \delta\phi_{r(l)}(y_l)}\ \left(A_0\otimes \cdots \otimes A_n\right)\Big|_{(\phi_0,\dots,\phi_n)=0}\ .
  \end{align*}
  The $F^{\textup{vac}}_{n,G}$ can be written as integrals in momentum space
  \begin{align*}
	F^{\textup{vac}}_{n,G}(U,\vec Z)&= \int dP\;\prod_{l\in E(G)} \textup e^{-p_l^0(u_{r(l)}-u_{s(l)})+i\vec p_l(\vec z_{s(l)}-\vec z_{r(l)})}\hat D_+^\textup{vac}(p_l) \hat \Psi(-P,P)\\
	&=\int d\vec P\; \prod_{l\in E(G)}\left( \textup e^{-\omega_{\vec p_l}(u_{r(l)}-u_{s(l)})+i\vec p_l(\vec z_{s(l)}-\vec z_{r(l)})} \frac{1}{2\omega_{\vec p_l}}\right)\hat \Psi(-P,P)\Big|_{p^0_l=\omega_{\vec p_l}}
  \end{align*}
  where $\omega_{\vec p_l}=\sqrt{\vec p_l^2+m^2}$. By Proposition \ref{prop:smooth} (in Appendix \ref{sec:appb}) we know that $\hat \Psi(-P,P)$ is rapidly decreasing in the forward lightcone, 
  and since $\supp\hat D_+^{\textup{vac}}  \subset H_m$ with 
  \begin{align*}
	H_m=\{p\in \RR^4: p_0^2-\vec p^2=m^2,p_0>0\}\subset J^+
  \end{align*}
  the above integral converges absolutely since by assumption $u_{r(l)}-u_{s(l)}>0$. Therefore we can make use of Proposition \ref{prop:contourdeform} from  Appendix \ref{sec:appb} to obtain the estimate
  \begin{align*}
	\abs{F^{\textup{vac}}_{n,G}(u_1,\vec x_1;\ldots;u_n,\vec x_n)}\le c \; \textup e^{-mr},\qquad r=\sum_{l\in G}\sqrt{\left(u_{r(l)}-u_{s(l)}\right)^2+\abs{\vec x_{r(l)}-\vec x_{s(l)}}^2}\ .
  \end{align*}
  Since the graph $G$ is connected, i.e.\ every vertex can be reached from $(u_0,\vec x_0)=0$, we can use
  \begin{align*}
   r& =\sum_{l\in G}\sqrt{ \left(u_{r(l)}-u_{s(l)}\right)^2+\abs{\vec x_{r(l)}-\vec x_{s(l)}}^2} \ge  \max_{i\in \{1,\ldots,n\}} \sqrt{u_i^2+\abs{\vec x_i}^2}\\
   &\ge \sqrt{\frac{1}{n} \sum_{i=1}^n u_i^2+\abs{\vec x_i}^2}=\frac{1}{\sqrt{n}}\sqrt{\sum_{i=1}^n u_i^2+\abs{\vec x_i}^2}\ ,
  \end{align*}
  which yields 
  \begin{align*}
		\abs{F^{\textup{vac}}_{n,G}(u_1,\vec x_1;\ldots;u_n,\vec x_n)}\le c' \text e^{-\frac{m}{\sqrt{n}} r_e},\qquad r_e=\sqrt{\sum_{i=1}^n u_i^2+\abs{\vec x_i}^2}\ .
  \end{align*}
  This shows that $F_{n,G}^\text{vac}$ actually decays exponentially in every variable $(u_i,\vec x_i)$, instead of only in the difference variables. 
  Consequently, since on the algebra of Wick polynomials only finitely many graphs contribute to the sum,
  also the summed expression 
  \begin{align*}
	F_n^\textup{vac}(u_1,\vec x_1;\ldots;u_n ,\vec x_n)=\sum_{G\in\mathcal G_{n+1}^\textup{c}} \frac{1}{\text{Symm}(G)} F_{n,G}^{\textup{vac}}(u_1,\vec x_1;\ldots;u_n ,\vec x_n)
  \end{align*}
  is exponentially decaying in its variables and is thus integrable over $\mathcal S_n^{\infty}\times \RR^{3n}$.
\end{proof}

Using the analyticity properties of the vacuum state we know that the co-cycle $U'_h(t)$, inserted as a right factor in the expectation value,  admits an analytic extension to the full upper half plane, using the first point of Proposition \ref{lemma:kmsanalytic} for the limiting case $\beta\to +\infty$. In particular, the linear functional 
\begin{align*}
	\omega_{\text{vac}}^{h}(A)&=\lim_{\beta\to\infty}\frac{\omega_\text{vac}(\gamma(A)U'_h(i\beta))}{\omega_\text{vac}( U'_h(i\beta))}\\
	&=\sum_{n=0}^\infty (-1)^n\int_{\mathcal S_n^\infty} du_1\cdots du_n\;\omega_{\text{vac}}^\textup{c}\left(\gamma(A)\otimes  \alpha_{iu_1}(K'_h) \otimes \cdots \otimes \alpha_{iu_n}(K'_h)\right)
\end{align*}
exists and is positive. Theorem \ref{thm:adlimcond} implies that the adiabatic limit 
\begin{align*}
	&\omega^{I}_{\text{vac}}(A)=\mathrm{vH}-\lim_{h\to I}\omega_{\text{vac}}^{h}(A)\\
	&=\sum_{n=0}^\infty \int_{\mathcal S_n^\infty}dU \int_{\RR^{3n}} dX\;\omega_{\text{vac}}^\textup{c}\left(\gamma(A)\otimes \alpha_{iu_1,\vec x_1}(R)\otimes 
	\cdots \otimes \alpha_{iu_n,\vec x_n}(R)\right)
\end{align*}
with $R=\int dt \; \dot\chi^-(t) [\mathcal H_I(t,0)]_{\chi}$ exists. In particular we find that the function
\begin{align*}
  &t\mapsto	\omega^{I}_{\text{vac}}(A  \alpha_t^{I}(B))
\end{align*}
(where $I$ as before denotes the constant function $h=1$) has a bounded analytic continuation into the whole upper half plane, which characterizes  a ground state.
\subsection{Thermal equilibrium states}

In this section, we  show that KMS states of perturbatively constructed, massive scalar field theories exist for all $0<\beta<+\infty$. The proof of this fact is very similar to the case of the vacuum, as far as the perturbative expansions are concerned. A main difference arises in the investigation of the decay behavior due to the fact that the KMS state has contributions from positive and negative energies. The KMS condition and the exponential decay of the negative energy part turn out to be crucial to show the convergence of the state in the adiabatic limit.

\begin{theorem}\label{thm:kmsal}
  Let $\omega_{{\beta}}$ be the quasi-free KMS state of the Klein-Gordon field with mass $m>0$ and inverse temperature $0<\beta<\infty$ whose translation invariant two-point function is
  \begin{align}
	\omega_\beta(\Phi(x)\Phi(y))=D_+^{\beta}(x-y), \quad D_+^{\beta}(x)=\frac{1}{(2\pi)^3}\int dp\;  \textup e^{-i(p_0 x^0-\vec p\vec x)}\frac{\varepsilon(p_0)\delta(p^2-m^2)}{1-\textup e^{-\beta p_0}} \label{eq:kms2pf2}\ .
	\end{align}
  Then the connected correlation function
  \begin{align*}
		F_n^\beta(u_1,\vec x_1;\ldots; u_n,\vec x_n)&=\omega_{\beta}^\textup{c}\left(A_0\otimes \alpha_{iu_1,\vec x_1}(A_1)\otimes \cdots \otimes \alpha_{iu_n,\vec x_n}(A_n)\right)
  \end{align*}
  for $A_0,\ldots,A_n\in\mathfrak A(\mathcal O)$ with $\mathcal O\subset B_R\subset \RR^4$ decays exponentially in spatial directions
  \begin{align*}
   \abs{F_n^\beta(u_1,\vec x_1;\ldots;u_n,\vec x_n)}\le c \; \text e^{-\frac{m}{\sqrt{n}}r_e},\qquad r_e=\sqrt{\sum_{i=1}^n \abs{\vec x_i}^2}\ ,
  \end{align*}
  uniformly for $r_e>2R$ and $(u_1,\ldots,u_n)\in \beta\mathcal S_n$.
\end{theorem}

\begin{proof}
  We proceed in the same manner as for the vacuum state. To this end we write 
  \begin{align*}
	F_n^\beta(u_1,\vec x_1;\ldots; u_n,\vec x_n)&=\omega_{\beta}^\textup{c}\left(A_0\otimes \alpha_{iu_1,\vec x_1}(A_1)\otimes \cdots \otimes \alpha_{iu_n,\vec x_n}(A_n)\right)\\
	&=\sum_{G\in\mathcal G_{n+1}^\textup{c}} \frac{1}{\textup{Symm}(G)}F_{n,G}^\beta(u_1,\vec x_1;\ldots ;u_n,\vec x_n)
  \end{align*}
  with
  \begin{equation}	
	F_{n,G}^\beta(u_1,\vec x_1;\ldots ;u_n,\vec x_n)=\prod_{i<j} (\Gamma_2^{ij})^{l_{ij}} \left(A_0\otimes \alpha_{iu_1,\vec x_1}(A_1)\otimes \cdots \otimes \alpha_{iu_n,\vec x_n}(A_n)\right)
	\Big|_{(\phi_0,\dots, \phi_n)=0}\nonumber
  \end{equation}
  where the $\Gamma_2^{ij}$ are the functional differential operators from the proof of Proposition~\ref{lemma:kmsanalytic}. The differential operator is now re-written in terms of the two-point function
  \begin{align*}
	\Gamma_+^{ij}&=\int dx\;dy\;D_+^\beta(x^0-y^0,\vec x-\vec y) \frac{\delta^2}{\delta\phi_i(x)\delta\phi_j(y)}\ .
  \end{align*}
  By switching to a product over the lines $l\in E(G)$ of the graph $G$ we find 
  \begin{align*}
	F_{n,G}^\beta(U,\vec Z)&= \int dP\;\prod_{l\in E(G)} \textup e^{p_l^0(u_{s(l)}-u_{r(l)})+i\vec p_l(\vec z_{s(l)}-\vec z_{r(l)})}\hat D_+^\beta(p_l) \hat \Psi(-P,P)\\
	&=\int dP\; \prod_{l\in E(G)} \frac{\text e^{i\vec p_l(\vec z_{s(l)}-\vec z_{r(l)})}\left( \lambda_+(p_l)+\lambda_-(p_l)\right)}{2\omega_{l}\left(1-\text e^{-\beta\omega_{l}}\right)} \hat \Psi(-P,P)
  \end{align*}
  with
  \begin{equation}\nonumber	
         \lambda_+(p_l)=\text e^{\omega_{l}(u_{s(l)}-u_{r(l)})}\delta(p^0_l-\omega_{l})\ ,\qquad 	\lambda_-(p_l)=\text e^{-\beta\omega_{l}}\text e^{\omega_{l}(u_{r(l)}-u_{s(l)})}\delta(p_l^0+\omega_{l})\ ,
  \end{equation}
  with $\omega_l\equiv\omega_{p_l}$ in analogy to the calculation for the vacuum state. The difference here is that $\hat D_+^\beta$ is not purely supported  in the 
  forward lightcone, it has positive and negative mass-shell part $\lambda_\pm$. The functional derivatives are, again, given by
  \begin{align}
  \Psi(X,Y)&=\prod_{l\in E(G)}  \frac{\delta^2}{\delta \phi_{s(l)}(x_l) \delta\phi_{r(l)}(y_l)}\ \left(A_0\otimes \cdots \otimes A_n\right)\Big|_{(\phi_0,\dots,\phi_n)=0}\ . \label{eq:kmsfunctionalder}
  \end{align}
  Due to the fact, that the integration momenta $p_l$ can lie in both the forward and backward lightcone, we cannot use the same argumentation as in the case of the vacuum state. 
  In order to prove the convergence of the integral we   will show that all negative energy parts $\lambda_-$  are actually exponentially decreasing.

  To this end, we use the KMS condition in the original function
  \begin{align*}
	F_n^\beta(u_1,\vec x_1;\ldots; u_n,\vec x_n)&=\omega_{\beta}^\textup{c}\left(A_0\otimes \alpha_{iu_1,\vec x_1}(A_1)\otimes \cdots \otimes \alpha_{iu_n,\vec x_n}(A_n)\right)
  \end{align*}
  together with the identification $(u_0,\vec x_0)=0$ to rearrange the time-translations in imaginary directions:
  \begin{align*}
    &\omega_{\beta}^\textup{c}\left(A_0\otimes \alpha_{iu_1,\vec x_1}(A_1)\otimes \cdots \otimes \alpha_{iu_n,\vec x_n}(A_n)\right)\\
    &=\omega_{\beta}^\textup{c}\Big(\alpha_{iu_m,\vec x_m}(A_m) \otimes \cdots \otimes \alpha_{iu_n,\vec x_n}(A_n) \otimes \alpha_{i\beta} (A_0)\otimes \alpha_{i(u_1+\beta),\vec x_1}(A_1)\otimes \cdots\\
    &\qquad \cdots \otimes \alpha_{i(u_{m-1}+\beta),\vec x_{m-1}}(A_{m-1})\Big)\ .
  \end{align*}
  The equality holds irrespective of the choice $m\in\{1,\ldots,n\}$. The non-trivial point is made now: There exists an $m\in\{1,\ldots,n\}$ such that $u_{m}-u_{m-1}\ge \frac{\beta}{n+1}$. We simply rename all of the variables to 
  \begin{align*}
	F_{n,G}^\beta(U,\vec X)={F'}_{n,G}^\beta(V,\vec Y)=\omega_{\beta}^\textup{c}\left(\alpha_{iv_0,\vec y_0}(B_0)\otimes \alpha_{iv_1,\vec y_1}(B_1)\otimes \cdots \otimes \alpha_{iv_n,\vec y_n}(B_n)\right)
  \end{align*}
  where $B_0=A_m,B_1=A_{m+1},\ldots,B_n=A_{m-1}$, 
 \begin{align}
	\begin{pmatrix} 
	  v_0\\ \vdots \\ v_{n-m} \end{pmatrix} = \begin{pmatrix}u_m\\ \vdots \\ u_n\end{pmatrix},\qquad  \begin{pmatrix}v_{n-m+1} \\ \vdots \\ v_n \end{pmatrix}=\begin{pmatrix} u_0+\beta \\ \vdots \\ u_{m-1}+\beta 
	\end{pmatrix}\ ,\label{eq:coordch}
  \end{align}
  and the similar relabeling is done for the spatial variables $\vec y_i$ with respect to $\vec x_i$. Now the analogous derivation for $F'$ yields 
  \begin{align*}
	&{F'}_{n,G}^\beta(v_0,\vec y_0;\ldots;v_n,\vec y_n)=\int dP\prod_{l\in E(G)} \frac{\text e^{i\vec p_l(\vec y_{s(l)}-\vec y_{r(l)})}\left( \lambda_+(p_l)+\lambda_-(p_l)\right)}{2\omega_{l}
	\left(1-\text e^{-\beta\omega_{l}}\right)} \hat\Psi_B(-P,P) \ ,
  \end{align*}
  with
  \begin{align*}	
	&\lambda_+(p_l)=\text e^{\omega_{l}(v_{s(l)}-v_{r(l)})}\delta(p^0_l-\omega_{l}),\qquad 	\lambda_-(p_l)=\text e^{\omega_{l}(v_{r(l)}-v_{s(l)}-\beta)}\delta(p_l^0+\omega_{l})\ ,
  \end{align*}
  where now $v_0\le v_1\le \ldots \le v_{n}$ and $\Psi_B$ is the functional derivative from equation \eqref{eq:kmsfunctionalder} in which the $(A_i)$ are replaced by $(B_i)$. 

  We expand the products of the sum of $\lambda_\pm$ by replacing  every line $l\in E(G)$ by either a line $l_+$ or $l_-$, to which we associate the factors $\lambda_\pm$, and summing over all 
  possibilities to distribute pluses and   minuses on all lines in $G$. This is done by introducing a function 
  \begin{align*}
	\varepsilon: E(G)\to \{+,-\}
  \end{align*}
  that associates signs to all the lines in the graph. Denoting $E_\pm(G)=\{l\in G:\varepsilon(l)=\pm\}$ we find 
  \begin{align*}
	&{F'}_{n,G}^\beta(v_0,\vec y_0;\ldots;v_n,\vec y_n)\\
	=&\sum_{\varepsilon}\int dP\;\prod_{l_+\in E_+(G)}\left[ \frac{\text e^{i\vec p_{l_+}(\vec y_{s(l_+)}-\vec y_{r(l_+)})} }{2\omega_{l_+}\left(1-\text e^{-\beta\omega_{l_+}}\right)} \lambda_+(p_{l_+})\right] \times \\
        &\quad	\times \prod_{l_-\in E_-(G)}\left[ \frac{\text e^{i\vec p_{l_-}(\vec y_{s(l_-)}-\vec y_{r(l_-)})} }{2\omega_{l_-}\left(1-\text e^{-\beta\omega_{l_-}}\right)} \lambda_-(p_{l_-})\right]\hat \Psi_B(-P,P). 
  \end{align*}
  Now we estimate the largest difference between the $v_i$ 
  \begin{align*}
	\max_{i<j} (v_j-v_i)=v_{n-1}-v_0=\beta +u_{m-1}-u_m=\beta-(u_m-u_{m-1})\le \underbrace{\frac{\beta}{n+1} }_{=: c_\beta}<\beta\ .
  \end{align*}
  Thus we rewrite
  \begin{align*}
	 \text e^{-\omega_{l_-}\beta } \text e^{\omega_{l_-}(v_{r(l_-)}-v_{s(l_-)})}& =\text e^{-\omega_{l_-}(\beta -c_\beta)} \text e^{\omega_{l_-}(v_{r(l_-)}-v_{s(l_-)}-c_\beta)}\\
	 &=\text e^{-\omega_{l_-}\frac{n\beta}{n+1}} \underbrace{\text e^{\omega_{l_-}(v_{r(l_-)}-v_{s(l_-)}-c_\beta)}}_{\le 1}\ ,
  \end{align*}
  which shows the claim that the integrand of ${F'}_{n}^\beta$ decays fast in the momentum variables associated to lines $l_-$. The remaining integration variables (those associated to $l_+$) are 
  located in the forward lightcone, in   which $F'^\beta_n$ is rapidly decreasing. This implies that 
  \begin{align*}
	\prod_{l_-\in E_-(G)} \text e^{-\omega_{l_-}(v_{r(l_-)}-v_{s(l_-)})} \hat \Psi_B(-P,P)\Big|_{p_{l_+}^0=\omega_{l_+},\; p_{l_-}^-=-\omega_{l_-}} 
  \end{align*}
  is rapidly decreasing in all spatial momenta $\vec P=(\vec p_1,\ldots,\vec p_{E(G)})$. We use the geometric series
  \begin{align*}
		\frac{1}{1-\textup e^{-\beta \omega}}=\sum_{k=0}^\infty \text e^{-\beta k \omega}
  \end{align*}
  to rewrite the integrand
  \begin{align*}
    &\sum_{\varepsilon} \prod_{l_+\in E_+(G)} \frac{1}{2\omega_{l_+}} \frac{\text e^{i\vec p_{l_+}(\vec y_{s(l_+)}-\vec y_{r(l_+)})-\omega_{l_+}(v_{r(l_+)}-v_{s(l_+)})}}{1-\text e^{-\beta\omega_{l_+}}} \times \\
	 &\qquad \times \prod_{l_-\in E_-(G)} \frac{1}{2\omega_{l_-}}\frac{\text e^{i\vec p_{l_-}(\vec y_{s(l_-)}-\vec y_{r(l_-)})-\omega_{l_-}(	c_\beta-(v_{r(l)}-v_{s(l)}))} 
	 \text e^{-\omega_{l_-}\frac{n\beta}{n+1}}}{1-\text e^{-\beta\omega_{l_-}}}\\
	&=\sum_{\varepsilon} \sum_{\vec k\in\NN_0^{\abs{E(G)}}}\prod_{l_+\in E_+(G)} \frac{1}{2\omega_{l_+}} \text e^{i\vec p_{l_+}(\vec y_{s(l_+)}-\vec y_{r(l_+)})-\omega_{l_+}(v_{r(l_+)}-v_{s(l_+)}+\beta k_{l_+})} \times \\
	 &\qquad \times \prod_{l_-\in E_-(G)} \frac{1}{2\omega_{l_-}}\text e^{i\vec p_{l_-}(\vec y_{s(l_-)}-\vec y_{r(l_-)})-\omega_{l_-}(	c_\beta-(v_{r(l)}-v_{s(l)})+\beta k_{l_-})} \text e^{-\omega_{l_-}\frac{n\beta}{n+1}}\ .
  \end{align*}
  Hence the function $F'$ is of the form
  \begin{align*}
		&{F'}_{n,G}^\beta(v_0,\vec y_0;\ldots;v_n,\vec y_n)\\
	&=\sum_{\varepsilon}  \sum_{\vec k\in\NN_0^{\abs{E(G)}}} \int  d\vec P\;\prod_{l_+ \in E_+ (G)} \frac{1}{2\omega_{l_+}}\text e^{i\vec p_{l_+}(\vec y_{s(l_+)}-\vec y_{r(l_+)})} 
	\text e^{-\omega_{l_+}(v_{r(l_+)}-v_{s(l_+)}+\beta k_{l_+})} \times \\
	&\qquad \times \prod_{l_-\in E_-(G)} \frac{1}{2\omega_{l_-}}\text e^{i\vec p_{l_-}(\vec y_{s(l_-)}-\vec y_{r(l_-)})}\text e^{-\omega_{l_-}(	c_\beta-(v_{r(l)}-v_{s(l)})+\beta k_{l_-})}  \; \Xi(\vec P)\ ,
  \end{align*}
  with
  \begin{align*}	
	&\Xi(\vec P)=\prod_{l_-\in E_-(G)}\text e^{-\omega_{l_-} \frac{n\beta}{n+1}}\hat \Psi_B(-P,P)\Big|_{p_{l_+}^0=\omega_{l_+},\; p_{l_-}^0=-\omega_{l_-}}\ . 
  \end{align*}
  By the above argumentation, $\Xi(\vec P)$ is rapidly decreasing in all its variables. Fixing the sign-function $\varepsilon$ and a multi-index $\vec n$, we can use 
  Proposition \ref{prop:contourdeform} from Appendix \ref{sec:appb} to   find the estimate
  \begin{align*}
	F'^\beta_{n,G,\varepsilon,\vec k}(v_0,\vec y_0;\ldots,v_n\vec y_n) \le c\;  \prod_{\substack{l\in E(G)}}\text e^{-m\sqrt{\abs{\vec x_{\partial l}}^2+(\beta k_l)^2}}\ ,
  \end{align*}
  where $\vec x_{\partial l}=\vec x_{r(l)}-\vec x_{s(l)}$. In this estimate we used the fact, that the $v_i$ range only over a finite interval and the differences 
  \begin{align*}
	c_\beta-(v_{r(l)}-v_{s(l)})\ge 0
  \end{align*}
  are bounded from below by zero. The sum over $k_l$ yields 
  \begin{align*}
	\sum_{k=0}^\infty \text e^{-m\sqrt{q^2+(\beta k)^2}}& =\sum_{\beta n< q} \text e^{-m\sqrt{q^2+(\beta k)^2}}+ \sum_{\beta n\ge q} \text e^{-m\sqrt{q^2+(\beta k)^2}}\\
	\le&\frac{q}{\beta} \text e^{-mq}+ \frac{\text e^{-mq}}{1-\text e^{-m\beta}}\le c' \; \text e^{-mq}
  \end{align*}
  for $q>0$. This implies that 
  \begin{align*}
	F'^\beta_{n,G,\epsilon}(V,\vec Y)&=\sum_{\vec n \in \NN^{E(G)}}F'^\beta_{n,G,\varepsilon,\vec k}(V,\vec Y)\le cc'\; \prod_{l\in E(G)}\text e^{-m\sqrt{\abs{\vec x_{\partial l}}^2}}\\
	&\le c'' \text e^{-\frac{m}{\sqrt{n}}r_e},\qquad r_e=\sqrt{\sum_{k=0}^n \abs{\vec x_k}^2} \ ,
  \end{align*}
  by the same means as in the case of the vacuum state. The exponential decay for $F_{n,G}^\beta$, i.e.\
  \begin{align*}
	\abs{	{F}_{n,G}^\beta(u_1,\ldots,u_n,\vec x_1,\ldots,\vec x_n)}\le c' \; \text e^{-\frac{m}{\sqrt{n}}r_e},\qquad r_e=\sqrt{\sum_{k=1}^{n} \abs{\vec x_k}^2}\ ,
  \end{align*}
  follows by the simple coordinate change in equation \eqref{eq:coordch}, thus $F_{n,G}^\beta$ decays exponentially in all its variables. The same decay properties hold for $F_{n}^\beta$, 
  which is the sum over all connected graphs of   $F_{n,G}^\beta$ divided by the symmetry factor of $G$. This proves the assertion.
\end{proof}

As in the vacuum case we can exploit the analytic properties of the KMS state $\omega_\beta$ to show that the limiting state obeys the KMS condition by using 
Proposition \ref{lemma:kmsanalytic}. We find an explicit formula for the adiabatic limit of the state $\omega'_{\beta}{}^{h}$
\begin{align*}
	\omega^{I}_{\beta}(A)&=\lim_{h\to I}	\omega'_{\beta}{}^{h}(A)\\
	&=\sum_{n=0}^\infty (-1)^n\int_{\beta\mathcal S_n} dU\int_{\RR^{3n}} dX\;\omega_{\beta}^\textup{c}\left(\gamma(A)\otimes \alpha_{iu_1,\vec x_1}(R)\otimes \cdots \otimes \alpha_{iu_n,\vec x_n}(R)\right)\ ,
\end{align*}
with
\begin{align*}	
	R&=\int dt\;  \dot\chi^-(t)\left[\mathcal H_I (t,0)\right]_{\chi}\ .
\end{align*}
The limit exists, defines a state on $\mathfrak A_{[I]}$, and the function
\begin{align*}
  &t\mapsto	\omega^{I}_{\beta}(A \alpha_t^{I}(B))
\end{align*}
has an analytic continuation into the strip $S_\beta$ and is continuous on the boundary with the value
\begin{align*}
	\omega_{\beta}^{I}(A \alpha_{t+i\beta}^{I}(B))=\omega_\beta^{I}(\alpha_{t}^{I}(B)A)\ .
\end{align*}
In order to prove these statements one has simply to replace the limiting KMS state $\omega_\beta^{I}$ with the ones on the finite volume $\omega_{\beta}^{h}$ in the proof of Theorem \ref{thm:analyticcont}. Since the arising integrands are absolutely integrable, we can exchange the limits in the integrations and obtain the desired statements. In particular we find that, due to the fact that the free KMS state is invariant under all spacetime translations and spatial rotations, so is the interacting state $\omega_\beta^{I}$ in the adiabatic limit.

\section{Conclusion}
In this work we were concerned with the construction of KMS states in perturbative renormalized relativistic QFT. The traditional approaches to construct states in perturbative QFT by 
starting from the corresponding state of the free theory in the limit $t\to-\infty$ lead, in the case of positive temperature, to infrared divergences as emphasized by Steinmann in \cite{steinmannfinite}. 
These divergences occur even in the massive case, see e.g.\ \cite{altherr}. Their physical origin may be traced back to the change in the asymptotic behavior in time as analyzed by  Bros and Buchholz in \cite{bb}. 

Thus we chose another approach by using ideas from the realm of quantum statistical mechanics. For this to achieve we had to close a gap between the description of QFT and QM as 
dynamical systems, which originates from the fact that non-trivial local interactions in relativistic QFTs in four dimensions are too singular to be restricted to Cauchy surfaces. Therefore, the  formal application of 
quantum mechanical perturbation theory leads to spurious UV divergences (the St\"{u}ckelberg UV-divergences \cite{stueckelberg}), that appear even after renormalization.

To avoid these difficulties, we exploited the validity of the time-slice axiom in causal perturbation theory \cite{cf} and embedded the algebra of interacting fields into the algebra of the free theory, restricted to a time-slice
which may be understood as a thickened Cauchy surface. We showed that  the interacting dynamics $\alpha_t^I$ differs from the free dynamics $\alpha_t$ locally by a well defined co-cycle $U_h(t)$, 
thus we obtain a UV regular interaction picture. We therefore can apply techniques from quantum statistical mechanics \cite{brob} and construct KMS states for the dynamics with a spatial cutoff $h$.

As a last step the spatial cutoff $h$ had to be removed (adiabatic or thermodynamic limit). We showed that this is always possible if the KMS state of the free theory has sufficiently good spatial clustering 
properties. We show that this assumption is satisfied for the free massive scalar field.  As a byproduct we also obtained a new proof for the existence of the vacuum state. We always formulated our arguments for four-dimensional Minkowski space, but all arguments are actually valid independent of the dimension. In two-dimensional spacetime, for polynomial interactions, our methods are not needed, because of the existence of a sufficiently 
large algebra of time zero fields.  

Our results are obtained within formal perturbation theory. Conceptually, however, the methods are not restricted to perturbation theory, and it would be interesting to explore whether they can be applied in the 
sense of constructive field theory, for instance in $\phi^4_3$. 

The non-zero mass of the theory was crucial for the proof of existence, since, in general, the correlation functions of massless theories exhibit a too slow decay at spatial infinity. The consequences of this are 
observed in many applications of the massless theories at positive temperature and are sometimes referred to as \emph{the IR problem} of perturbative QFT at positive temperature, see e.g.\ \cite{altherr}. 
A solution of this problem is to use the so-called \emph{thermal mass} term, that arises, due to finite renormalization of the interaction terms, as a mass term of the free theory. The idea has already been 
used in the context of QED and QCD at positive temperature \cite{lvw}, an interpretation of this method in the general context of perturbative QFT in the present setting is given in \cite{lindner}.
In the massless $\phi^4$-theory it originates from the fact that the transformation from the $\star$-product with the vacuum two-point function to the $\star$-product at finite inverse temperature $\beta$ involves a 
finite transformation which corresponds to a change of the Wick ordering prescription, 
\begin{equation}
:\!\phi^4\!\!:_{\mathrm{vac}}=:\!\phi^4\!\!:_{\beta}+6c(\beta):\!\phi^2\!\!:_{\beta}+6c(\beta)
\end{equation} 
where 
\begin{equation}
c(\beta)=(D_+^{\beta}-D_+^{\mathrm{vac}})(0)=\frac{1}{12\pi^2\beta^2}\ .
\end{equation}
The coefficient of the quadratic term then can be used as a mass term, and one obtains convergent expressions for the KMS state. But since $D_+^{\beta}$ is not analytic in $m^2$ at $m=0$, the arising series is no longer a power series in the coupling constant.

We believe that our approach provides a basis for extending powerful methods from quantum statistical mechanics to perturbative QFT. In particular, one might apply these methods to problems of 
non-equilibrium thermodynamics in QFT, as treated e.g.\ in \cite{berges,buchm}. One may also test whether structural properties of the interacting KMS state predicted from an axiomatic approach, 
such as the relativistic KMS condition \cite{bbrelkms}, K\"{a}llen-Lehmann type representations \cite{bb2} and time asymptotics \cite{bb} can be verified. We also hope that the problem of phase transitions can be 
addressed by our method, yielding a bridge between the quantum statistical approach and the arguments relying on the effective potential of quantum field theory.

Another potential application of our results is on the treatment of bound states in quantum field theory, including in particular the derivation of the Lamb shift. Here a lot of progress was made during the last years in the so-called nonrelativistic QED (see e.g. \cite{Spohn}). Approaches to relativistic quantum mechanics typically suffer from inconsistencies since the concept of particles in QFT is dependent on the interaction. Our approach is completely consistent 
with the principles of QFT, but, unfortunately, still based on formal perturbation theory.  

Our method has a formal similarity with the Schwinger-Keldysh approach. There the expectation values of time-ordered products of fields $A_i$ are obtained from the expectation value of the 
free theory similar to the Gell-Mann and Low formula  
\begin{align*}
	\omega_\beta^{I} (TA_1(x_1)\dots A_n(x_n))=\frac{1}{Z} \omega_\beta\!\left(T_CA_1(x_1)\dots A_n(x_n) \exp(\!-i\!\int_{C}\!\!\!dz\! \int\! \!d^3\vec x\; \mathcal H_I(z,\vec x) h(\vec x))\!\!\right)\!,
\end{align*}
but where time ordering is replaced by ordering along the contour $C$ depicted in figure \ref{fig:schwinger}, see \cite{lvw} for more details to  this formalism.
\begin{figure}[htb]
  \begin{center}
	\includegraphics*[scale=1]{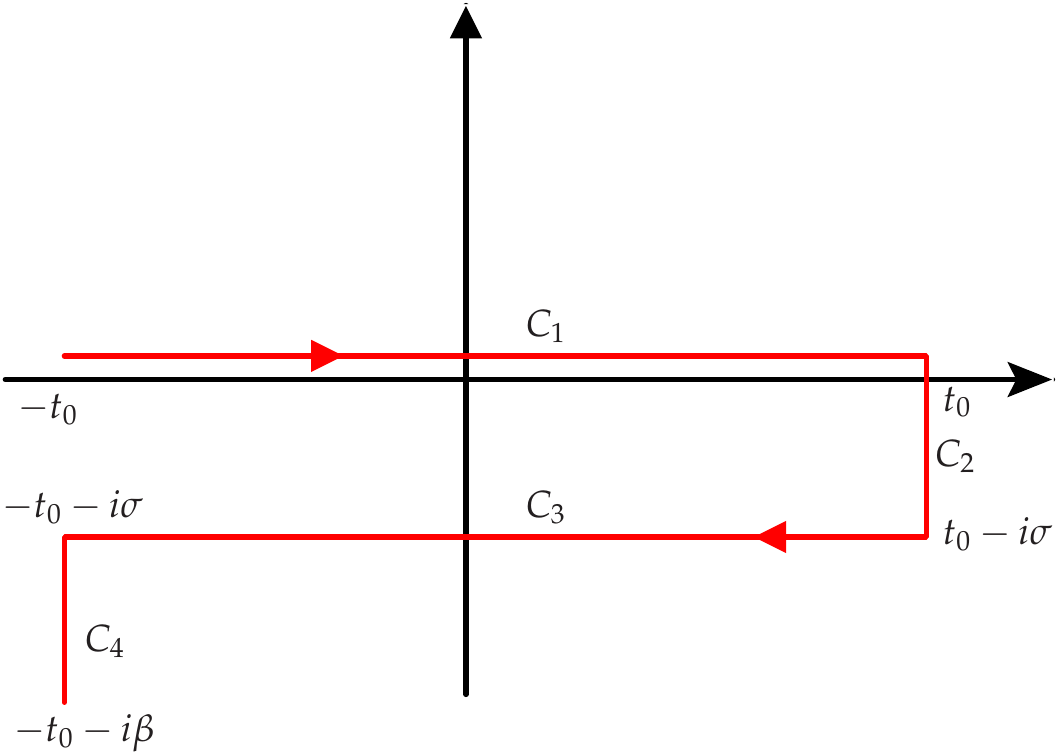}
  \end{center}
  \caption{The Schwinger-Keldysh contour $C=C_1\cup C_2\cup C_3\cup C_4$ in the complex time-plane. \label{fig:schwinger}}
\end{figure}
The right hand side of this formula might be written as a path integral. 
The attempt to construct the interacting KMS state by the adiabatic limit $g\to 1$ as discussed in section \ref{sec:2} corresponds to the contour in the limit $t_0\to\infty$. 
By choosing a contour $C$ which oscillates several times parallel to the real axis, one would formally obtain the thermal Wightman functions of the interacting field, if the contour 
reaches the time-arguments of the fields in the correct order.

The choice $t_0=\epsilon$  for the contour $C$ (together with the restriction to observables that are supported in $\Sigma_\epsilon$) corresponds to the idea of exploiting the time-slice axiom, 
where the time cutoff $\chi$ is replaced by the characteristic function of the interval $[-\epsilon,\epsilon]$. This avoids the aforementioned IR divergences completely, but would generate 
additional UV divergences at the boundaries $t=\pm\epsilon$. Our formalism may be thus understood as a Schwinger-Keldysh formalism with a smeared, but finitely extended contour.

\section*{Acknowledgments}
We thank the referees of the paper for many valuable hints. One of the authors (F.~L.) gratefully acknowledges financial support from the Konrad-Adenauer-Stiftung. 
\appendix
\section{Proof of the KMS condition}\label{sec:kms}
In order to show that the quasi-free state $\omega_\beta$, defined by
\begin{align}
	\omega_\beta(\Phi(x)\Phi(0))=D_+^\beta(x)=\frac{1}{(2\pi)^3}\int \frac{d^3\vec p \;\text e^{i\vec p \vec x}}{2\omega_{p}(1-\text e^{-\beta \omega_p})}   
	\left(\text e^{-i\omega_{p}x^0} +\text e^{i\omega_p (x^0+i\beta)} \right)\ ,\label{eq:kmsfree}
\end{align}
satisfies the analytic conditions from Definition \ref{def:kms} we will use an \emph{off-shell formulation} of the algebra of the Wick polynomials of the free field. 
For an introduction we refer to \cite{paqftbuch} or \cite{keller,frerejz}. In this formulation the algebra of Wick polynomials $\mathcal A$ is constructed by (non necessarily linear) functionals over the 
space of classical field configurations $\phi\in C^\infty(M)$, $M=\RR^4$, which are not subject to any field equation. A local field $A$, smeared with a test function $f\in\mathcal D(M)$, is in this formalism given by  
\begin{align*}
	A(f)(\phi)= \int d^4x \; f(x) a(\phi(x),\partial \phi(x),\dots) \ ,
\end{align*}
where $a$ is a polynomial in $\phi$ and its derivatives. 
The most general observable is given by sums of 
\begin{align}
	F_T(\phi)=\int d^4x_1\cdots d^4x_n\; T(x_1,\ldots ,x_n) \phi(x_1)\cdots \phi(x_n) \label{eq:microcausal}
\end{align}
with a distribution of compact support $T\in \mathcal E'(M^n)$ with the following singularity structure
\begin{align}
	\WF(T)\subset \left\{(x_1,\ldots,x_n|p_1,\ldots,p_n)\in \dot T^*M^n:\sum_{k=1}^n p_k=0\right\}.\label{eq:wfmicrocausal}
\end{align}
Here $\dot T^*$ denotes the cotangent bundle of a manifold, with the zero section removed. 
These functional are called microcausal functionals. It is evident  that $A(f)$ is contained in this set as one replaces $T$ by a sum of $\delta$-distributions and its derivatives 
multiplied with a test function $f$. The product between such functionals is declared by a functional differential operator (a $\star$-product in the sense of deformation quantization)
\begin{align}
	(F \star G)(\phi) = \textup e^{\Gamma_2}(F\otimes G)(\phi)\label{eq:starprod}
\end{align}
with the differential operator 
\begin{align*}
	\Gamma_2(F\otimes G)=\int d^4x\; d^4y\; \omega_2(x,y) \frac{\delta F}{\delta\phi(x)}\otimes \frac{\delta G}{\delta\phi(y)}\ . 
\end{align*}
Here $\omega_2$ is the two-point function of some quasi-free Hadamard state $\omega$ over the algebra of Wick polynomials $\mathfrak A$. The $^*$-algebra which is $\star$-generated by 
functionals \eqref{eq:microcausal} is denoted by $\mathcal A_{\omega}$. The connection to the Hilbert space formalism of QFT is the following: There is a $^*$-homomorphism 
$\pi_{\omega}:\mathcal A_{\omega}\to \mathfrak A$ into the the algebra of Wick-polynomials, given by
\begin{align}
	\pi_{\omega}(F_T)=\int d^4x_1\cdots d^4x_n\; T(x_1,\ldots,x_n) \normord{\Phi(x_1)\cdots \Phi(x_n)}_{\omega_2}
\end{align}
where the Wick-ordering $\normord{\cdot}_{\omega_2}$ within $\mathfrak A$ has been  done with respect to the two-point function $\omega_2$. Note that if the two-point function of another quasifree Hadamard state were to be put in 
$\Gamma_2$, the corresponding $^*$-homomorphism would result in a locally quasi-equivalent representation \cite{Verch}. 

The representation $\pi_{\omega}$ is not faithful, since elements of the form $P\Phi(x)$ vanish on $\mathfrak A$ whereas they don't on $\mathcal A_{\omega}$.\footnote{$P$ 
denotes the Klein-Gordon differential operator $P=\square+m^2$ on $M$.} This is a consequence  of the off-shell setting in $\mathcal A_{\omega}$. It has been shown that the 
kernel of $\pi_{\omega}$ is the ideal $\mathcal I$ generated by elements of the form
\begin{align*}
	F_{P_iT}(\phi)=\int d^4 x_1\cdots d^4x_n\; (\square_{x_i}+m^2)T(x_1,\ldots,x_n) \phi(x_1)\ldots \phi(x_n) \in \mathcal A_{\omega_2}\ .
\end{align*}
The algebra $\quotient{\mathcal A_{\omega}}{\mathcal I}$ is called the \emph{on-shell algebra}. Notice that the quasi-free state $\omega$ on $\mathfrak A$  is found as the evaluation functional on $\mathcal A_{\omega}$
\begin{align*}
	A\mapsto A(0)=\omega(\pi_{\omega}(A)),\qquad A\in\mathcal A_{\omega}.
\end{align*}
\begin{proposition}\label{lemma:kmsanalytic}
	Let $\omega_\beta$ be a quasi-free KMS state on $\mathfrak A$ with respect to the time-evolution $\alpha_t$ and with inverse temperature $0<\beta \le +\infty$. Then, for every $A_1,\ldots,A_n\in\mathfrak A$,  the functions
	\begin{align*}
	  (t_1,\ldots,t_n)\mapsto	\omega_\beta(\alpha_{t_1}(A_1)\cdots \alpha_{t_n}(A_n))
	\end{align*}
        have an analytic continuation into 
         \begin{align*}
	  \mathfrak T^n_\beta&=\{(z_1,\ldots z_n)\in \CC^n: 0<  \Im(z_j-z_i) < \beta  \quad \forall 1\le i<j\le n\}\ .
	\end{align*}
         Moreover, for $0<\beta<\infty$ and $k\in \{1,\ldots,n\}$ it holds that
	\begin{align*}
		&\omega_\beta(\alpha_{t_1}(A_1)\cdots \alpha_{t_k}(A_k)\alpha_{t_{k+1}+i\beta}(A_{k+1})\cdots \alpha_{t_n+i\beta}(A_n))\\
		&=\omega_\beta(\alpha_{t_{k+1}}(A_{k+1})\cdots \alpha_{t_n}(A_n)\alpha_{t_1}(A_1)\cdots \alpha_{t_k}(A_k))\ .
	\end{align*}
\end{proposition}

\begin{proof}
  We prove the Proposition for the case $0<\beta< \infty$. The case $\beta=+\infty$ is proven along the same lines, except from the KMS conditions. It is assumed that the two-point function $D_+^\beta$ of $\omega_\beta$, 
  defined by equation  \eqref{eq:kmsfree}, is used in the $\star$-product \eqref{eq:starprod} and the resulting off-shell algebra is denoted by $\mathcal A_\beta$. A multiple product of observables  can be written as 
  \begin{align*}
		(A_1\star \cdots \star A_n)(\phi)&= \prod_{1\le i<j \le n}\textup e^{\Gamma_2^{ij}} (A_1\otimes \cdots \otimes A_n)\Big|_{\phi_1=\dots=\phi_n=\phi}\ ,
  \end{align*}
  with		
  \begin{align*}		
           \Gamma_2^{ij}&= \int d^4x\; d^4y\; D_+^\beta(x-y) \frac{\delta^2}{\delta \phi_i(x)\delta\phi_j(y)}\ ,
  \end{align*}
  using the Leibniz rule of differential calculus on functionals. Here the $n$-fold tensor product of functionals are considered as functionals in $n$ field configurations $\phi_1,\dots,\phi_n$ .

  The time-translations $\alpha_t$ on $\mathcal A_{\omega}$ are simply given by $\alpha_t(A)(\phi)=A(\phi_{-t})$ where $\phi_{-t}(x)=\phi(x^0+t,\vec x)$ in agreement with  the 
  notation in \eqref{eq:timetranslation}. The expectation   value in the KMS state is then the evaluation at $\phi=0$, and it holds that
  \begin{align*}
	\omega_\beta(\alpha_{t_1}(A_1)\cdots \alpha_{t_n}(A_n))&=\left(\alpha_{t_1}(A_1)\star \cdots \star \alpha_{t_n}(A_n)\right)(\phi=0)\\
	&=\prod_{1\le i<j \le n}\textup e^{\Gamma_2^{ij}} (\alpha_{t_1}(A_1)\otimes \cdots \otimes \alpha_{t_n}(A_n))\Big|_{\phi_1=\dots= \phi_n=0}\\
	&= \prod_{1\le i<j \le n}\textup e^{\Gamma_2^{ij}(t_i,t_j)} (A_1\otimes \cdots \otimes A_n)\Big|_{\phi_1=\dots= \phi_n=0}\ ,
  \end{align*}
  with
  \begin{align*}	
	\Gamma_2^{ij}(t_i,t_j)& = \int d^4x\; d^4y\; D_+^\beta(x^0-y^0+t_i-t_j,\vec x-\vec y) \frac{\delta^2}{\delta \phi_i(x)\delta\phi_j(y)}\ .
  \end{align*}
  As  $t\mapsto D_\beta(t,\vec x)$ has an analytic continuation into $-S_\beta=\{z\in\CC:-\beta<\Im(z)< 0\}$  the function 
  \begin{align*}
	(t,t')\mapsto \textup e^{\Gamma_2^{ij}(t,t')}(A_1\otimes \cdots \otimes A_n)
  \end{align*}
  has an analytic continuation to $\{(z_1,z_2)\in \CC^2:-\beta < \Im(z_1-z_2) <0\}$ for all $A_1,\ldots,A_n\in \mathcal A_{\beta}$. Thus for the full expectation value we obtain that the map  
  \begin{align*}
	(t_1,\ldots,t_n) \mapsto \prod_{1\le i<j \le n}\textup e^{\Gamma_2^{ij}(t_i,t_j)} (A_1\otimes \cdots \otimes A_n)\Big|_{\phi_1=\dots=\phi_n=0}
  \end{align*}
  has an extension into 
  \begin{align*}
	\mathfrak T^n_\beta&=\{(z_1,\ldots,z_n)\in\CC^n: 0 <\Im(z_j-z_i) <\beta \quad \forall 1\le i<j\le n\}\ .
  \end{align*}
  The KMS conditions for the expectation values 
  \begin{align*}
	\omega_\beta(\alpha_{t_1}(A_1)\cdots  \alpha_{t_n}(A_n))=\prod_{1\le i<j \le n}\textup e^{\Gamma_2^{ij}(t_i,t_j)} (A_1\otimes \cdots \otimes A_n)\Big|_{\phi_1=\dots= \phi_n=0}
  \end{align*}
  follow directly from the fact that $\Gamma_2^{ij}(t_i,t_j)=\Gamma_2^{ji}$ if $t_i-t_j=-i\beta$.
\end{proof}
\section{Propositions for the adiabatic limit}\label{sec:appb}
The following propositions are used to show that  the connected correlation functions have a uniform exponential decay in spatial and imaginary time directions. 
Parts of this are already well-known in the case of the vacuum two-point function $\omega_\text{vac}(\Phi(x)\Phi(y))$, see e.g.\ \cite{bogolubovintro}. For our purposes however, a more general statement has to be proven.
\begin{proposition}\label{prop:contourdeform}
  Let $f\in\mathcal D'(\RR^{4})$ with $\supp(f)\subset B_R\subset \RR^4$. Then the functions
  \begin{align*}
	I_f(x_0,\vec x)&=\int \frac{d^3 \vec p}{2\omega_p}\;\textup e^{-i(x_0\omega_p-\vec p\vec x)}  \hat f(\omega_p,\vec p),\qquad \omega_p=\sqrt{\vec p^2+m^2}\ ,\\
	I^b_f(x_0,\vec x)&=\int \frac{d^3 \vec p}{2\omega_p}\;\textup e^{-i(x_0\omega_p-\vec p\vec x)} \textup e^{-b\omega_p}  \hat f(-\omega_p,\vec p),\qquad b>0\ ,
  \end{align*}
  have an analytic continuation into the lower half plane $\CC_-\times \RR^3$, and for $m>0$ it holds that
  \begin{align*}
	\abs{I_f(-iu,\vec x)}\le c \; \textup e^{-m r}\ ,\quad \abs{I_f^b(-iu,\vec x)}\le c \; \textup e^{-m r}\ ,\qquad r=\sqrt{u^2+\vec x^2}\ ,
  \end{align*}
  uniformly for $r\ge2R$.
\end{proposition}

\begin{proof}
  The domain of analyticity of $I_f$ and $I_f^b$ is obvious due to the fact that  $\hat f$ is a polynomially bounded function while the remaining integrands for
  \begin{align*}
	I_f(-iu,\vec x):  \quad &\frac{\textup e^{i\vec p \vec x -u\sqrt{\vec p^2+m^2}}}{2\sqrt{\vec p^2+m^2}}\\
	I_f^b(-iu,\vec x):\quad   &\frac{\textup e^{i\vec p \vec x -u\sqrt{\vec p^2+m^2}}}{2\sqrt{\vec p^2+m^2}}\text e^{-b\omega_p}
  \end{align*}
  decay exponentially for $u>0$ and $b>0$. The following steps will be discussed for $I_f$ only, the case for $I_f^b$ is identical until further notice. Using the identity
  \begin{align*}
	\frac{1}{2\pi}\int dk \frac{\textup e^{i k u}}{k^2+\omega^2}=\frac{\textup e^{-\omega u}}{2\omega},\qquad u,\omega>0\ ,
  \end{align*}
  we can rewrite the  integral:
  \begin{align*}
	I_f(-iu,\vec x)&=\int d^3\vec p\;\frac{ \textup e^{i\vec p \vec x -u\sqrt{\vec p^2+m^2}}}{2\sqrt{\vec p^2+m^2}} \hat f\left(\sqrt{\vec p^2+m^2},\vec p\right)\\
	&=\frac{1}{2\pi}\int dp\; \frac{\text e^{i(up_0+\vec x\vec p)}}{p_0^2+\vec p^2+m^2} \hat f\left(\sqrt{\vec p^2+m^2},\vec p\right)\ .
  \end{align*}
  Without loss of generality, we choose the coordinates $\vec x= \vec nr\cos(\alpha)$ and $u=r\sin(\alpha)$ with $\vec n=(1,0,0)$  and $0< 2\alpha<\pi$. Hence, $r=\sqrt{u^2+\vec x^2}$. 
  The following change in the momentum variables is helpful:
  \begin{align*}
	k_0&=p_0\sin(\alpha)+p_1\cos(\alpha),\quad 	k_1=p_0\cos(\alpha)-p_1\sin(\alpha),	\quad k_{2/3}=p_{2/3}\ .
  \end{align*}
  The integral is of the form
  \begin{align*}
	I_f(-iu,\vec x)&=\frac{1}{2\pi}\int dk\; \frac{\text e^{ik_0 r}\hat f\left(\omega(k),\vec p(k)\right)}{k_0^2+k_1^2+k_2^2+k_3^2+m^2} =\frac{1}{2\pi}\int dk\; \frac{\text e^{ik_0 r}\hat f
	\left(\omega(k),\vec p(k)\right)}{k_0^2+\vec k^2+m^2}
  \end{align*}
  where 
  \begin{align*}
     \omega(k)&=\sqrt{(k_0\cos(\alpha)-k_1\sin(\alpha))^2+k_2^2+k_3^2+m^2}\ ,\\
     \vec p(k)&=(k_0\cos(\alpha)-k_1\sin(\alpha),k_2,k_3)\ .
   \end{align*}
   We replace the integration in the $k_0$-variable by a contour integration in the upper half plane. By the Paley-Wiener theorem \cite{strichartz} we know that for $(\omega,\vec p)$ in the upper half plane $\CC_+^4$:
   \begin{align*}
     \abs{\hat f(\omega,\vec p)}\le c \;\textup e^{R\sqrt{\abs{\omega}^2+\abs{\vec p}^2}}\ .
   \end{align*}
   Thus  the integrand will exponentially decay for values $r>R$. We thus look at the contour integral
   \begin{align*}
	\frac{1}{2\pi}\int_Cdz\; \frac{\text e^{iz r}}{z^2+ \vec k^2+m^2} \hat f(\omega,\vec p)\ ,
   \end{align*}
   where the dependence of $\omega$ and $\vec p$ on the variables $k_0=z$ and $\vec k$ is suppressed in the notation of $\hat f$. We see that the integrand has a pole at 
   \begin{align*}
	z=i\sqrt{\vec k^2+m^2}\ .
   \end{align*}
   Furthermore, since the principal square root in $\omega(z,\vec k)$ has a branch cut on the negative real axis, we get a branch cut along a vertical axis, starting from
   \begin{align*}
     z=k_1 \tan(\alpha)+i \frac{\sqrt{k_2^2+k_3^2+m^2}}{\cos(\alpha)}\ .
   \end{align*}
  \begin{figure}[htb]
	 \def\svgwidth{\textwidth}
          \input{contour.tex}
          \caption{The integration contour $C$. The semicircle has to be extended to infinite size and the orientation is positive. \label{fig:ccontour}}
  \end{figure}
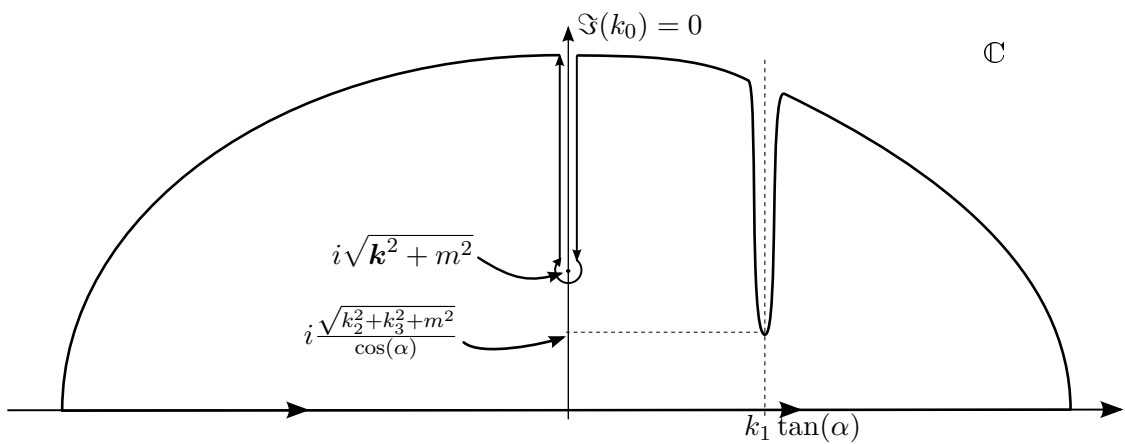
  Thus we choose a contour that avoids both the pole and the branch cut (see figure \ref{fig:ccontour}) such that the contour integral vanishes due to the exponential decay of the integrand:
  \begin{align*}
    0&=	\frac{1}{2\pi}\int_C dz\; \frac{\text e^{iz r}}{z^2+ \vec k^2+m^2 } \hat f(\omega,\vec p)\\
   &=I_f(x_0,\vec  x)+\frac{1}{2\pi}\oint\limits_{\text{pole}} dz\frac{\text e^{iz r}}{z^2+ \vec k^2+m^2} \hat f(\omega,\vec p)\\
   &\qquad \qquad+\frac{1}{2\pi}\int\limits_\text{branch} dz\;\frac{\text e^{iz r}}{z^2+ \vec k^2+m^2 } \hat f(\omega,\vec p)\ .
  \end{align*}
  The pole contour can be calculated using the residue theorem
  \begin{align*}
          &\frac{1}{2\pi}\oint\limits_{z=i\sqrt{\vec k^2+m^2}}dz\; \frac{\text e^{iz r}}{z^2+ \vec  k^2+m^2} \hat f(\omega,\vec p)\\
	&= i\; \text{Res}_{z=i\sqrt{\vec k^2+m^2}}\;\frac{\text e^{iz r}\hat f(\omega,\vec p)}{\left(z-i\sqrt{\vec k^2+m^2}\right)\left(z+i\sqrt{\vec k^2+m^2}\right)}\\
	&= \frac{\text e^{-r\sqrt{\vec k^2+m^2}}\hat f(\omega,\vec p)} {2\sqrt{\vec k^2+m^2}}\Big|_{k_0=i\sqrt{\vec k^2+m^2}}\ ,
  \end{align*}
  thus the full pole contribution to $I_f$ is
  \begin{align*}
    I_{\text{pole}}(-iu,\vec x)=\int d^3\vec k\; \frac{\text e^{-r\sqrt{\vec k^2+m^2}}\hat f(\omega,\vec p)} {2\sqrt{\vec k^2+m^2}}\Big|_{k_0=i\sqrt{\vec k^2+m^2}}\ .
  \end{align*}
  The branch cut contributes with 
  \begin{align*}
	\frac{i}{\cos\alpha}\int_{\Omega}^\infty d\tau \; \frac{\text e^{irk_1\tan(\alpha)} \text e^{-\frac{r\tau}{\cos(\alpha)}}}{z(\tau)^2+\vec k^2+m^2} \left(\hat f(\omega,\vec p)\big|_{k_0=z(\tau)+i\epsilon}-\hat f(\omega,\vec p)\big|_{k_0=z(\tau)-i\epsilon}\right)
  \end{align*}
  where $z(\tau)=k_1\tan(\alpha)+i\frac{\tau}{\cos(\alpha)}$ and $\Omega=\sqrt{k_2^2+k_3^2+m^2}$. The arguments of $\hat f$ in this parametrization of the branch cut are given as
  \begin{align*}
	\omega(k_0=z(\tau),k_i)&=\pm i\sqrt{\tau^2-k_2^2-k_3^2-m^2}\\
	\vec p(k_0=z(\tau),k_i)&=(i\tau ,k_2,k_3)\ .
  \end{align*}
  In particular,  $\hat f$ does not depend on $k_1$ on the branch cut. We invoke the $k_1$-integration to find
  \begin{align*}
	\int 	dk_1\int_{\Omega}^\infty d\tau \;\frac{\text e^{irk_1\tan(\alpha)} \text e^{-\frac{r\tau}{\cos(\alpha)}}}{z(\tau)^2+ \vec  k^2+m^2} (\hat f(\omega,\vec p)_+-\hat f(\omega,\vec p)_-) \ .
  \end{align*}
  We replace the $k_1$-integration by a contour-integration along a semi-circle in the upper half plane, where the integrand falls off exponentially (as $r>R$):
  \begin{align*}
	\int_C 	dw\int_{\Omega}^\infty d\tau \;\frac{\text e^{irw\tan(\alpha)} \text e^{-\frac{r\tau}{\cos(\alpha)}}}{z(\tau)^2+ w^2+ k_2^2+k_3^2+m^2} (\hat f(\omega,\vec p)_+-\hat f(\omega,\vec p)_-) .
  \end{align*}
  The fact that $\hat f$ does not depend on $k_1$ on the branch cut implies that the only contribution to the integral comes from the poles of the integrand, which are located at 
  \begin{align*}
	w_{\text{pole}}=-i\tau\sin(\alpha)\pm \cos(\alpha)\sqrt{\tau^2-k_2^2-k_3^2-m^2}\ .
  \end{align*}
  By assumption $0<2\alpha<\pi$, thus both poles lie in the lower half plane and the full contour integral vanishes. Moreover the integrand falls off exponentially in the 
  upper half plane, such that the branch cut does not contribute to $I_f$ at all.

  Hence, we have  $I_f(-iu,\vec x)=I_\text{pole}(-iu,\vec x)$ and this contribution can be estimated by
  \begin{align*}
	\abs{I_f(-iu,\vec x)}&=\abs{ 	\int d^3\vec k\; \frac{\text e^{-r\sqrt{\vec k^2+m^2}}\hat f(\omega(k),\vec p(k))} {2\sqrt{\vec k^2+m^2}}\Big|_{k_0=i\sqrt{\vec k^2+m^2}}}\\
        &\le c\int d^3\vec k \frac{\text e^{-r\sqrt{\vec k^2+m^2}}\text e^{R\sqrt{\abs{\omega(k)}^2+\abs{\vec p(k)}^2}}} {2\sqrt{\vec k^2+m^2}}\Big|_{k_0=i\sqrt{\vec k^2+m^2}}\ .
  \end{align*}
  The values of $\omega(k)$ and $\vec p(k)$ at the pole $k_0=z=i\sqrt{\vec k^2+m^2}$ are 
  \begin{align*}
	\omega(k)&=\sqrt{\vec k^2+m^2}\sin(\alpha)-ik_1\cos(\alpha),\quad p_1(k)=i\sqrt{\vec k^2+m^2}\cos(\alpha)-k_1\sin(\alpha)\ .
  \end{align*}
  Thus under the square root of the integrand we have
  \begin{align*}
    &\abs{\omega(k)}^2+\abs{\vec p(k)}^2=\vec k^2+m^2 +k_1^2+k_2^2+k_3^2+m^2=2(\vec k^2+m^2) \ .
  \end{align*}
  This implies that the integral decays exponentially in $\abs{\vec x}=r$ uniformly in $r$ as $r>2R$, since:
  \begin{align*}
	\abs{I_f(-iu,\vec x)}&\le c \int d^3\vec k\;\frac{\text e^{-r\sqrt{\vec k^2+m^2}}\text e^{\sqrt{2}R\sqrt{\vec k^2+m^2}}} {2\sqrt{\vec k^2+m^2}}\\
	&\le c \;\text e^{-(r-\sqrt{2}R)m} \int d^3\vec k\;\frac{\text e^{-(r-\sqrt{2}R)(\sqrt{\vec k^2+m^2}-m)}} {2\sqrt{\vec k^2+m^2}}\\
	&\le c \; \text e^{-(r-\sqrt{2}R)m} \underbrace{\int d^3\vec k\;\frac{\text e^{-(2-\sqrt{2})R(\sqrt{\vec k^2+m^2}-m)}} {2\sqrt{\vec k^2+m^2}}}_{<\infty}\ .
  \end{align*}
  The same argumentation can be done for $I_f^b$:
  \begin{align*}
	\abs{I^b_f(-iu,\vec x)}&=\abs{I_\text{pole}^b(-iu,\vec x)}\le c \int d^3\vec k \frac{ \text e^{(\sqrt{2}R-b-r) \sqrt{\vec k^2+m^2}}}{2\sqrt{\vec k^2+m^2}} \ .
  \end{align*}
  This proves the claim.
\end{proof}

The other proposition concerns the singular directions of the functional derivatives that 
appear in the expansion of the truncated vacuum expectation values.
\begin{proposition}\label{prop:smooth}
  Define for $A_0,\ldots,A_n\in \mathfrak A$ the compactly supported distribution 
  \begin{align*}
	 \Psi(x_1,\ldots,x_k,y_1,\ldots,y_k)=\prod_{l=1}^k\frac{\delta^2}{\delta\phi_{s(l)}(x_l)\delta \phi_{r(l)}(y_l)} (A_0\otimes \cdots \otimes A_n)\Big|_{\phi_0=\dots=\phi_n=0}
  \end{align*}
  where $s,r:\{1,\ldots,k\}\to \{0,\ldots,n\}$ such that $s(l)<r(l)$ for all $l\in\{1,\ldots,k\}$ and let $\hat\Psi$ denote its Fourier transform. Then 
  \begin{align*}
    (p_1,\ldots,p_k) \mapsto \hat\Psi(-p_1,\ldots,-p_k,p_1,\ldots p_k)
  \end{align*}
  is rapidly decreasing inside a neighborhood of the union of $k$-fold product of the closed forward lightcone $(V^+)^k$ with that of the backward lightcone $(V^-)^k$.
\end{proposition}

\begin{proof}
  Using the tensor product rule for wavefront sets (see \cite{hoermander}) and the fact that the functionals $A_i$ are microcausal (see the beginning of Appendix \ref{sec:kms}),  
  one finds that $\hat \Psi(-P,P)$ is rapidly decaying in every direction, except the cone defined by
  \begin{align*}
    \Big\{(p_1,\ldots ,p_{k})\in \dot T^*M^{k}: \sum_{\substack{l=1,\ldots,k\\ s(l)=m}} p_l- \sum_{\substack{l=1,\ldots,k\\ r(l)=m}}p_l =0,\quad m=0,\ldots,n\Big\}
  \end{align*}
  Assume that all of the momenta lie inside either the forward or backward lightcone. Taking the first condition ($m=0$) we see that that $\{l\in \{1,\ldots,k\}:r(l)=0\}=\emptyset$. 
 This implies
  \begin{align*}
	\sum_{\substack{l=1,\ldots,k\\ s(l)=0}}p_l=0 \quad \Longrightarrow \quad p_l=0  \quad \forall l\in \{1,\ldots,k\}:s(l)=0\ .
  \end{align*}
  But the set $\{l\in \{1,\ldots,k\}: s(l)=0\}$ contains in particular all indices $\{l\in\{1,\ldots,k\}: r(l)=1\}$. This information can be put into the next condition, $m=1$, which yields
  \begin{align*}
	\sum_{\substack{l=1,\ldots,k\\ s(l)=1}}p_l- \underbrace{\sum_{\substack{l=1,\ldots,k\\ r(l)=1}}p_l}_{=0} = 0\quad \Longrightarrow \quad p_l=0 \quad \forall l\in\{1,\ldots,k\}: s(l)=1
  \end{align*}
  as, once again, all the directions are contained in one of the lightcones $V^+$, $V^-$. This can be iterated until $m=n$ with the result that all momenta $\{p_l:l=1,\ldots,k\}$ vanish, hence
  $(p_1,\ldots ,p_{k})\not\in \dot T^*M^{k}$.
\end{proof}
\bibliographystyle{plain}
\bibliography{thermalarxiv}
\end{document}

%% file: drawing-mod.tex
\begingroup%
  \makeatletter%
  \providecommand\color[2][]{%
    \errmessage{(Inkscape) Color is used for the text in Inkscape, but the package 'color.sty' is not loaded}%
    \renewcommand\color[2][]{}%
  }%
  \providecommand\transparent[1]{%
    \errmessage{(Inkscape) Transparency is used (non-zero) for the text in Inkscape, but the package 'transparent.sty' is not loaded}%
    \renewcommand\transparent[1]{}%
  }%
  \providecommand\rotatebox[2]{#2}%
  \ifx\svgwidth\undefined%
    \setlength{\unitlength}{761.16704102bp}%
    \ifx\svgscale\undefined%
      \relax%
    \else%
      \setlength{\unitlength}{\unitlength * \real{\svgscale}}%
    \fi%
  \else%
    \setlength{\unitlength}{\svgwidth}%
  \fi%
  \global\let\svgwidth\undefined%
  \global\let\svgscale\undefined%
  \makeatother%
  \begin{picture}(1,0.66303748)%
    \put(0,0){\includegraphics[width=\unitlength]{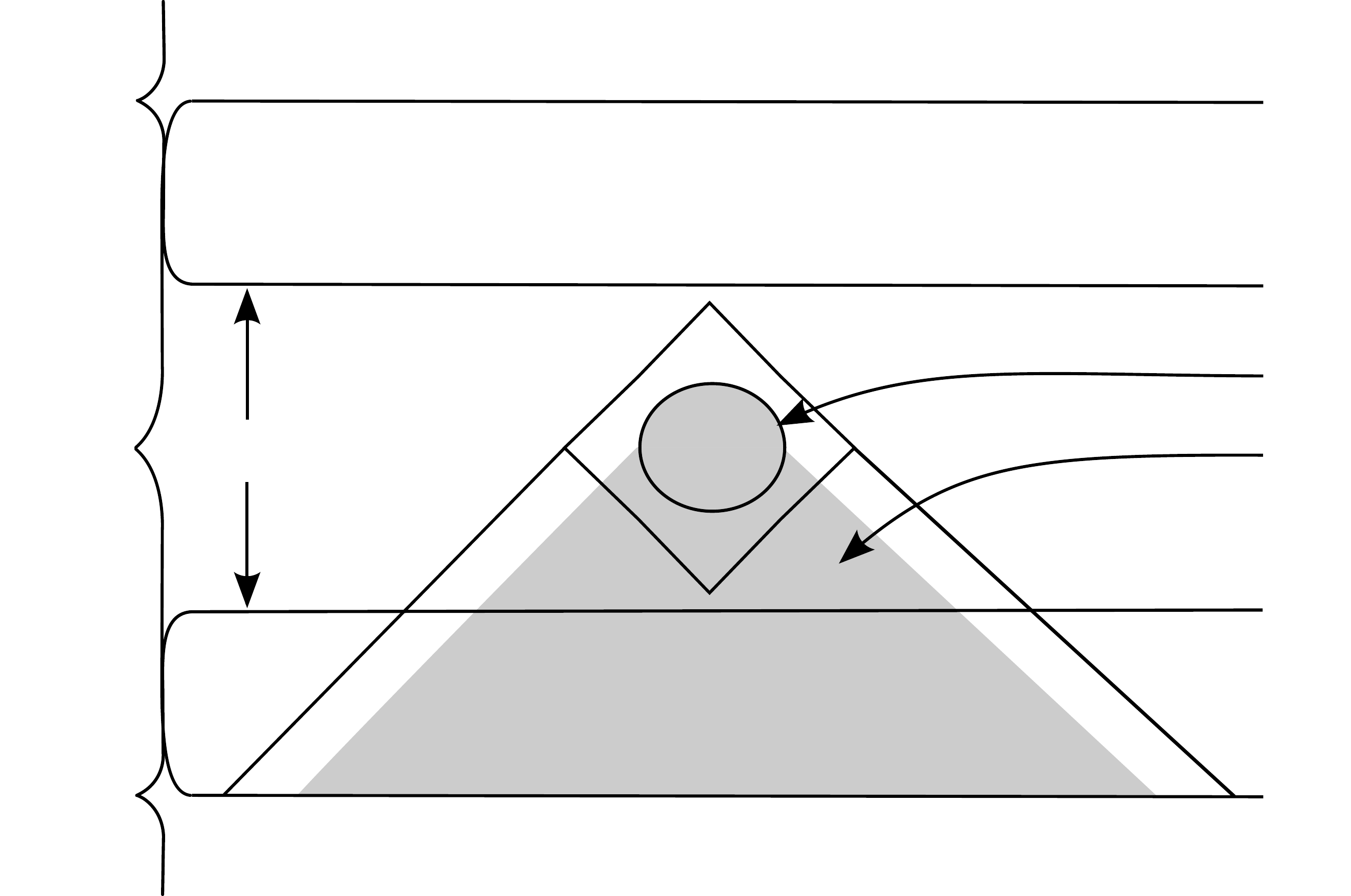}}%
    \put(0.95,0.38){\color[rgb]{0,0,0}\makebox(0,0)[lb]{\smash{$\supp(f)$}}}%
    \put(-0.045,0.58){\color[rgb]{0,0,0}\makebox(0,0)[lb]{\smash{$\supp(\chi_+)$}}}%
    \put(-0.03,0.32){\color[rgb]{0,0,0}\makebox(0,0)[lb]{\smash{$\supp(\chi)$}}}%
    \put(-0.045,0.07){\color[rgb]{0,0,0}\makebox(0,0)[lb]{\smash{$\supp(\chi_-)$}}}%
    \put(0.94994423,0.58){\color[rgb]{0,0,0}\makebox(0,0)[lb]{\smash{$x^0=2\epsilon$}}}%
    \put(0.94994423,0.44){\color[rgb]{0,0,0}\makebox(0,0)[lb]{\smash{$x^0=\epsilon$}}}%
    \put(0.94994423,0.20){\color[rgb]{0,0,0}\makebox(0,0)[lb]{\smash{$x^0=-\epsilon$}}}%
    \put(0.94994423,0.07){\color[rgb]{0,0,0}\makebox(0,0)[lb]{\smash{$x^0=-2\epsilon$}}}%
    \put(0.94994423,0.32){\color[rgb]{0,0,0}\makebox(0,0)[lb]{\smash{$C(f)$}}}%
    \put(0.17,0.315){\color[rgb]{0,0,0}\makebox(0,0)[lb]{\smash{$\Sigma_\epsilon$}}}%
  \end{picture}%
\endgroup%

%% file: cutoff.tex
\begingroup%
  \makeatletter%
  \providecommand\color[2][]{%
    \errmessage{(Inkscape) Color is used for the text in Inkscape, but the package 'color.sty' is not loaded}%
    \renewcommand\color[2][]{}%
  }%
  \providecommand\transparent[1]{%
    \errmessage{(Inkscape) Transparency is used (non-zero) for the text in Inkscape, but the package 'transparent.sty' is not loaded}%
    \renewcommand\transparent[1]{}%
  }%
  \providecommand\rotatebox[2]{#2}%
  \ifx\svgwidth\undefined%
    \setlength{\unitlength}{694.45625bp}%
    \ifx\svgscale\undefined%
      \relax%
    \else%
      \setlength{\unitlength}{\unitlength * \real{\svgscale}}%
    \fi%
  \else%
    \setlength{\unitlength}{\svgwidth}%
  \fi%
  \global\let\svgwidth\undefined%
  \global\let\svgscale\undefined%
  \makeatother%
  \begin{picture}(1,0.38924447)%
    \put(0,0){\includegraphics[width=\unitlength]{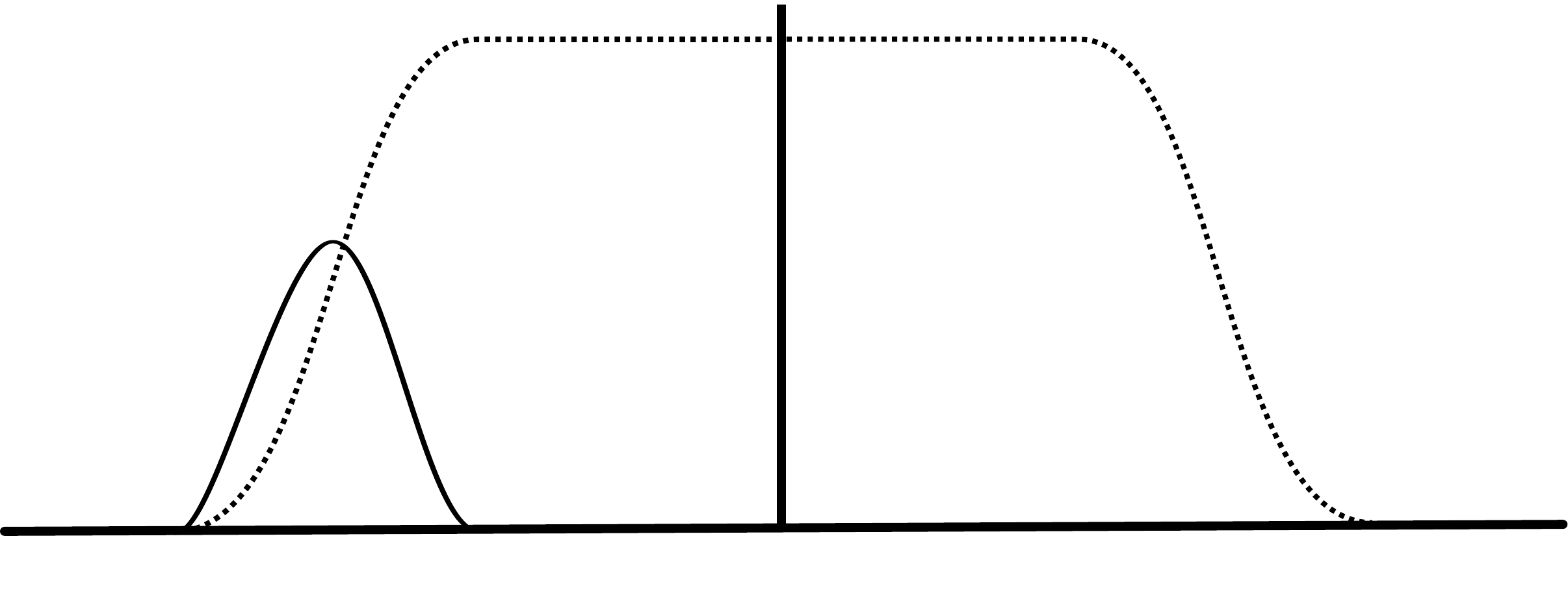}}%
    \put(0.45702617,0.00743654){\color[rgb]{0,0,0}\makebox(0,0)[lb]{\smash{$t=0$}}}%
    \put(0.77018432,0.29060466){\color[rgb]{0,0,0}\makebox(0,0)[lb]{\smash{$\chi$}}}%
    \put(0.24944501,0.19324541){\color[rgb]{0,0,0}\makebox(0,0)[lb]{\smash{$\dot\chi^-$}}}%
  \end{picture}%
\endgroup%

%% file: contour.tex
\begingroup%
  \makeatletter%
  \providecommand\color[2][]{%
    \errmessage{(Inkscape) Color is used for the text in Inkscape, but the package 'color.sty' is not loaded}%
    \renewcommand\color[2][]{}%
  }%
  \providecommand\transparent[1]{%
    \errmessage{(Inkscape) Transparency is used (non-zero) for the text in Inkscape, but the package 'transparent.sty' is not loaded}%
    \renewcommand\transparent[1]{}%
  }%
  \providecommand\rotatebox[2]{#2}%
  \ifx\svgwidth\undefined%
    \setlength{\unitlength}{912.05263672bp}%
    \ifx\svgscale\undefined%
      \relax%
    \else%
      \setlength{\unitlength}{\unitlength * \real{\svgscale}}%
    \fi%
  \else%
    \setlength{\unitlength}{\svgwidth}%
  \fi%
  \global\let\svgwidth\undefined%
  \global\let\svgscale\undefined%
  \makeatother%
  \begin{picture}(1,0.37056577)%
    \put(0,0){\includegraphics[width=\unitlength]{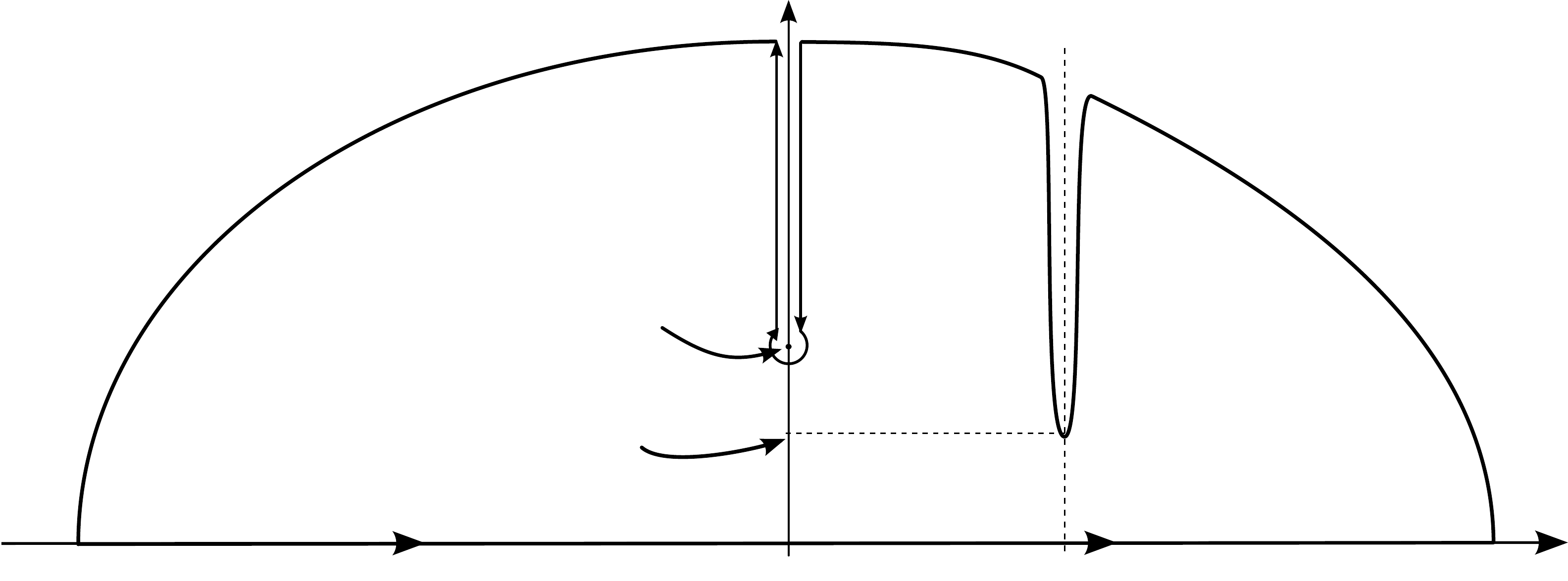}}%
    \put(0.51095605,0.3630714){\color[rgb]{0,0,0}\makebox(0,0)[lb]{\smash{$\Im(k_0)=0$}}}%
    \put(0.87326853,0.3380653){\color[rgb]{0,0,0}\makebox(0,0)[lb]{\smash{$\mathbbm C$}}}%
    \put(0.2911753,0.15611249){\color[rgb]{0,0,0}\makebox(0,0)[lb]{\smash{$i\sqrt{\vec k^2+m^2}$}}}%
    \put(0.26633139,0.08375863){\color[rgb]{0,0,0}\makebox(0,0)[lb]{\smash{$i\frac{\sqrt{k_2^2+k_3^2+m^2}}{\cos(\alpha)}$}}}%
    \put(0.66051038,0.00113583){\color[rgb]{0,0,0}\makebox(0,0)[lb]{\smash{$k_1\tan(\alpha)$}}}%
  \end{picture}%
\endgroup%